\documentclass{statsoc}

\usepackage{geometry}

\usepackage{etoolbox}

\makeatletter
\patchcmd{\@makecaption}
  {\parbox}
  {\advance\@tempdima-\fontdimen2} % decrease the width!
  {}{}
\makeatother

\usepackage{amsthm, amsmath, amssymb, bbm, booktabs, graphicx, makecell, multirow, textcomp, url, xcolor, adjustbox} 
\usepackage[figuresright]{rotating}
\usepackage[colorlinks,citecolor=blue]{hyperref}
\usepackage{setspace}

\theoremstyle{plain}  % default 
\newtheorem{thm}{Theorem}[section] 
\newtheorem{lem}[thm]{Lemma} 
\newtheorem{prop}[thm]{Proposition}

\theoremstyle{definition} 
\newtheorem{defn}{Definition}[section]

\theoremstyle{remark}

\newcommand{\real}{\mathbb{R}}
\newcommand{\bp}{\pmb{p}}
\newcommand{\bq}{\pmb{q}}

\newcommand{\bx}{\pmb{x}}
\newcommand{\by}{\pmb{y}}

\newcommand{\bF}{\pmb{F}}
\newcommand{\bhF}{\pmb{\hat{F}}}
\newcommand{\bhq}{\pmb{\hat{q}}}

\newcommand{\cA}{\mathcal{A}}
\newcommand{\cL}{\mathcal{L}}
\newcommand{\cO}{\mathcal{O}}
\newcommand{\cP}{\mathcal{P}}
\newcommand{\cQ}{\mathcal{Q}}
\newcommand{\cR}{\mathcal{R}}
\newcommand{\cS}{\mathcal{S}}
\newcommand{\cX}{\mathcal{X}}

\newcommand{\myE}{\mathbb{E}}
\newcommand{\myP}{\mathbb{P}}

\newcommand{\mys}{{\rm s}}
\newcommand{\myS}{{\rm S}}

\newcommand{\lest}{\leq_{\rm st}}
\newcommand{\st}{\preceq_{\rm st}}
\newcommand{\icx}{\preceq_{\rm icx}}

\newcommand{\xctr}{x_{\rm CTR}}  
\newcommand{\xhres}{x_{\rm HRES}}  
\newcommand{\mptb}{m_{\rm PTB}}   
\newcommand{\pzero}{p_{\rm ZERO}}   
\newcommand{\xptb}{x_{\rm PTB}}   

\newcommand{\IDRcw}{IDR$_{\rm cw}$}
\newcommand{\IDRicx}{IDR$_{\rm icx}$}
\newcommand{\IDRsbg}{IDR$_{\rm sbg}$}

\newcommand{\diff}{\, {\rm d}}
\newcommand{\one}{\mathbbm{1}}
\newcommand{\hsp}{\hspace{0.2mm}}

\newcommand{\orange}{\color{black}}

\newcommand{\green}{\color{black}}

\title[Isotonic Distributional Regression]{Isotonic Distributional Regression}
\author{Alexander Henzi and Johanna F.~Ziegel}
\address{University of Bern, Switzerland}
\email{alexander.henzi@stat.unibe.ch \ johanna.ziegel@stat.unibe.ch}
\author{Tilmann Gneiting}
\address{Heidelberg Institute for Theoretical Studies (HITS) and Karlsruhe Institute of \\ Technology (KIT), Germany}
\email{tilmann.gneiting@h-its.org}

\begin{document}

\begin{abstract} 
Isotonic distributional regression (IDR) is a powerful nonparametric
technique for the estimation of conditional distributions under order
restrictions.  In a nutshell, IDR learns conditional distributions
that are calibrated, and simultaneously optimal relative to
comprehensive classes of relevant loss functions, subject to
isotonicity constraints in terms of a partial order on the covariate
space.  Nonparametric isotonic quantile regression and {\green
  nonparametric isotonic} binary {\green regression} emerge as special
cases.  For prediction, we propose an interpolation method that
generalizes extant specifications under the pool adjacent violators
algorithm.  We recommend the use of IDR as a generic benchmark
technique in probabilistic forecast problems, as it does not involve
any parameter tuning nor implementation choices, except for the
selection of a partial order on the covariate space.  The method can
be combined with subsample aggregation, with the benefits of smoother
regression functions and {\green gains in} computational {\orange
  efficiency.}  {\green In a simulation study, we compare methods for
  distributional regression in terms of the continuous ranked
  probability score (CRPS) and $L_2$ estimation error, which are
  closely linked.}  In a case study on raw and postprocessed
quantitative precipitation forecasts from a leading numerical weather
prediction system, IDR is competitive with state of the art
techniques.

\end{abstract} 

\section{Introduction}  \label{sec:introduction} 
\footnotetext{\emph{This article is published as:} Henzi, A., Ziegel, J.F.~and Gneiting, T. (2021), Isotonic distributional regression. \emph{Journal of the Royal Statistical Society Series B}, \url{https://doi.org/10.1111/rssb.12450}.}
There is an emerging consensus in the transdisciplinary literature
that regression analysis should be distributional, with
\citet{Hothorn2014} arguing forcefully that 
\begin{quote}
[t]he ultimate goal of regression analysis is to obtain information
about the conditional distribution of a response given a set of
explanatory variables.
\end{quote}
Distributional regression marks a clear break from the classical view
of regression, which has focused on estimating the conditional mean of
the response variable in terms of one or more explanatory variable(s) or
covariate(s).  Later extensions have considered other functionals of
the conditional distributions, such as quantiles or expectiles
\citep{Koenker2005, Newey1987, SchulzeWaltrup2015}.  However, the
reduction of a conditional distribution to a single-valued functional
results in tremendous loss of information.  Therefore, from the
perspectives of both estimation and prediction, regression analysis
ought to be distributional.

In the extant literature, both parametric and nonparametric approaches
to distributional regression are available.  Parametric approaches
assume that the conditional distribution of the response is of a
specific type (e.g., Gaussian) with an analytic relationship between
the covariates and the distributional parameters.  Key examples
include statistically postprocessed meteorological and hydrologic
forecasts, as exemplified by \citet{Gneiting2005b},
\citet{Schefzik2013} and \citet{Vannitsem2018}.  In powerful
semi-parametric variants, the conditional distributions remain
parametric, but the influence of the covariates on the parameter
values is modeled nonparametrically, e.g., by using generalized
additive models \citep{Rigby2005, Klein2015, Umlauf2015} or modern
neural networks \citep{Rasp2018, Gasthaus2019}.  In related
developments, semiparametric versions of quantile regression
\citep{Koenker2005} and transformation methods \citep{Hothorn2014} can
be leveraged for distributional regression.

Nonparametric approaches to distributional regression include kernel
or nearest neighbor methods that depend on a suitable notion of
distance on the covariate space.  Then, the empirical distribution of
the response for neighboring covariates in the training set is used
for distributional regression, with possible weighting in dependence
on the distance to the covariate value of interest.  Kernel smoothing
methods and mixture approaches allow for absolutely continuous
conditional distributions \citep{Hall1999, Dunson2007, Li2008}.
Classification and regression trees partition the covariate space into
leaves, and assign constant regression functions on each leaf
\citep{Breiman1984}.  Linear aggregation via bootstrap aggregation
(bagging) or subsample aggregation (subagging) yields random forests
\citep{Breiman2001}, which are increasingly being used to generate
conditional predictive distributions, as proposed by
\citet{Hothorn2004} and \citet{Meinshausen2006}.

Isotonicity is a natural constraint in estimation and prediction
problems.  Consider, e.g., postprocessing techniques in weather
forecasting, where the covariates stem from the output of numerical
weather prediction (NWP) models, and the response variable is the
respective future weather quantity.  Intuitively, if the NWP model
output indicates a larger precipitation accumulation, the associated
regression functions ought to be larger as well.  Isotonic
relationships of this type hold in a plethora of applied settings.  In
fact, standard linear regression analysis rests on the assumption of
isotonicity, in the form of monotonicity in the values of the
covariate(s), save for changes in sign.

Concerning nonparametric regression for a conditional functional, such
as the mean or a quantile, there is a sizable literature on estimation
under the constraint of isotonicity.  The classical work of
\citet{Brunk1955}, \citet{Ayer1955}, \citet{vanEeden1958},
\citet{Bartholomew1959a, Bartholomew1959b} and \citet{Miles1959} is
summarized in \citet{Barlow1972}, \citet{RWD1988} and
\citet{deLeeuw2009}.  Subsequent approaches include Bayesian and
non-Bayesian smoothing techniques \citep[e.g.,][]{Mammen1991,
  Neelon2004, Dette2006, Shively2009}, and reviews are available in
\citet{Groeneboom2014} and \citet{Guntuboyina2018}.

In distributional regression it may not be immediately clear what is
meant by isotonicity, and the literature typically considers one
ordinal covariate only \citep[e.g.,][]{Hogg1965, Rojo2003,
  ElBarmi2005, Davidov2012}, with a notable exception being the work
of \citet{Moesching2020}, whose considerations allow for a real-valued
covariate.  In the general case of a partially ordered covariate
space, which we consider here, it is unclear whether semi- or
nonparametric techniques might be capable of handling monotonicity
contraints, and suitable notions of isotonicity remain to be
developed.

To this end, we assume that the response $Y$ is real-valued, and equip
the covariate space $\cX$ with a partial order $\preceq$.  Our aim is
to estimate the conditional distribution of $Y$ given the covariate
$X$, for short $\cL(Y \hsp | \hsp X)$, on training data, in a way that
respects the partial order, and we desire to use this estimate for
prediction.  Formally, a distributional regression technique generates
a mapping from $x \in \cX$ to a probability measure $F_x$, which
serves to model the conditional distribution $\cL(Y \hsp | \hsp X =
x)$.  This mapping is isotonic if $x \preceq x'$ implies $F_x \lest
F_{x'}$, where $\lest$ denotes the usual stochastic order, i.e., $G
\lest H$ if $G(y) \geq H(y)$ for $y \in \real$, where we use the same
symbols for the probability measures $G$, $H$ and their associated
conditional cumulative distribution functions (CDFs).  Equivalently,
$G \lest H$ holds if $G^{-1}(\alpha) \leq H^{-1}(\alpha)$ for $\alpha
\in (0,1)$, where $G^{-1}(\alpha) = \inf \{ y \in \real : G(y) \geq
\alpha \}$ is the standard quantile function \citep{Shaked2007}.

Useful comparisons of predictive distributions are in terms of proper
scoring rules, of which the most prominent and most relevant instance
is the continuous ranked probability score
\citep[CRPS;][]{Matheson1976, Gneiting2007a}.  We show that there is a
unique isotonic distributional regression that is optimal with respect
to the CPRS (Theorem \ref{thm:existence}), and refer to it as the {\em
  isotonic distributional regression} (IDR).  As it turns out, IDR is
a universal solution, in that the estimate is optimal with respect to
a broad class of proper scoring rules (Theorem
\ref{thm:universality}).  Classical special cases such as
nonparametric isotonic quantile regression and probabilistic
classifiers for threshold-defined binary events are nested by IDR.
Simultaneously, IDR avoids pitfalls commonly associated with
nonparametric distributional regression, such as suboptimal partitions
of the covariate space and level crossing \citep[p.~1167]{Athey2019}.

\begin{figure}[t]
\centering
\includegraphics[width = \textwidth]{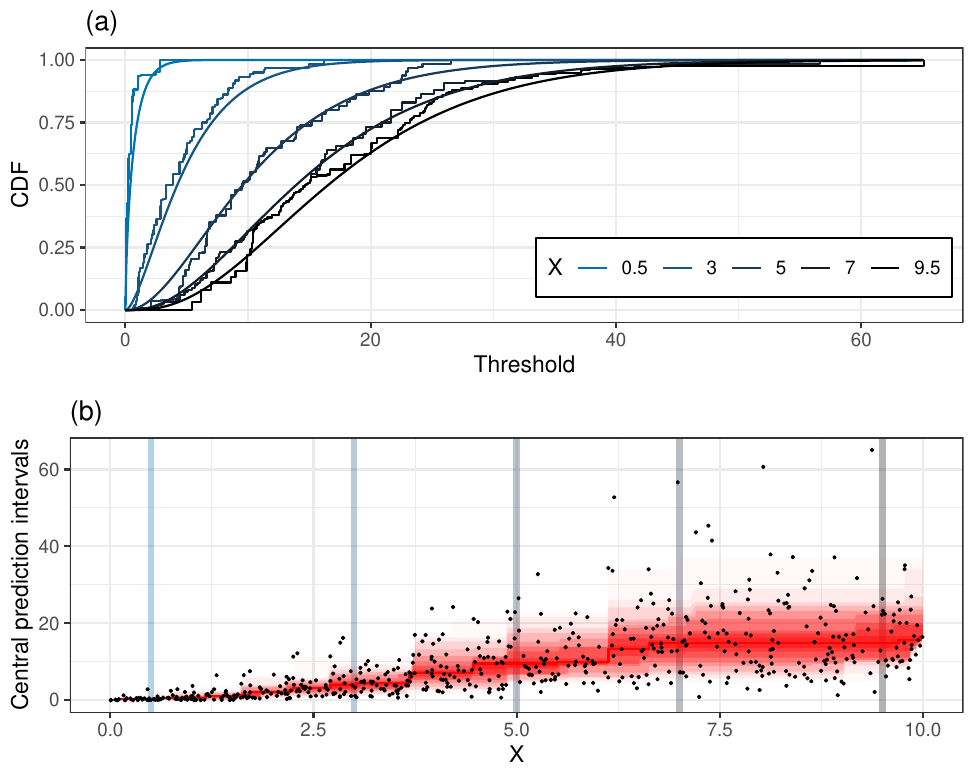}
\caption{Simulation example for a sample of size $n = 600$ from the
  distribution in \eqref{eq:sim}: (a) True conditional CDFs (smooth)
  and IDR estimates (step functions) for selected values of the
  covariate. (b) IDR estimated conditional distributions.  The shaded
  bands correspond to probability mass 0.10 each, with the darkest
  shade marking the central interval. Vertical strips indicate
  the cross-sections corresponding to the values of the covariate in 
  panel (a).  \label{fig:sim}}
\end{figure}

For illustration, consider the joint distribution of $(X,Y)$, where
$X$ is uniform on $(0, 10)$ and
\begin{equation}  \label{eq:sim} 
Y \mid X \sim \textrm{Gamma}(\textrm{shape} = \sqrt{X}, \, 
   \textrm{scale} = \min \{ \max \{ X, 1 \}, 6 \}), 
\end{equation}
so that $\cL(Y \hsp | \hsp X = x) \leq_{\rm st} \cL(Y \hsp | \hsp X =
x')$ if $x \leq x'$.  Figure \ref{fig:sim} shows IDR conditional CDFs
and quantiles as estimated on a training set of size $n = 600$.  IDR
is capable of estimating both the strongly right-skewed conditional
distributions for lower values of $X$ and the more symmetric
distributions as $X$ increases.  The CDFs are piecewise constant, and
they never cross each other.  The computational cost of IDR is of
order at least $\cO(n \log n)$ and may become prohibitive as $n$
grows.  However, IDR can usefully be combined with subsample
aggregation (subagging), much in the spirit of random forests
\citep{Breiman2001}, with the benefits of reduced computational cost
under large training samples, smoother regression functions, and
(frequently) improved predictive performance.

The remainder of the paper is organized as follows.  The
methodological core of the paper is in Section \ref{sec:IDR}, where we
prove existence, uniqueness and universality of the IDR solution,
discuss computational issues and asymptotic consistency, and propose
strategies for prediction.  In Section \ref{sec:orders} we turn to the
critical issue of the choice of a partial order on the covariate
space.  Section \ref{sec:simulation_study} reports on a comparative
simulation study that addresses both prediction and estimation, and
Section \ref{sec:case_study} is devoted to a case study on
probabilistic quantitative precipitation forecasts, with covariates
provided by the European Centre for Medium-Range Weather Forecasts
(ECMWF) ensemble system.  Precipitation accumulations feature
unfavorable properties that challenge parametric approaches to
distributional regression: The conditional distributions have a point
mass at zero, and they are continuous and right skewed on the positive
half-axis.  In a comparison to state-of-the-art methods that have been
developed specifically for the purpose, namely Bayesian Model
Averaging \citep[BMA;][]{Sloughter2007}, Ensemble Model Output
Statistics \citep[EMOS;][]{Scheuerer2014}, and Heteroscedastic
Censored Logistic Regression \citep[HCLR;][]{Messner2014}, the
(out-of-sample) predictive performance of IDR is competitive, despite
the method being generic, and being fully automatic once a partial
order on the covariate space has been chosen.

We close the paper with a discussion in Section \ref{sec:discussion},
where we argue that IDR provides a very widely applicable, competitive
benchmark in probabilistic forecasting problems.  The use of benchmark
techniques has been called for across application domains
\citep[e.g.,][]{Rossi2013, Pappenberger2015, BIS2016, Vogel2018}, and
suitable methods should be competitive in terms of predictive performance,
while avoiding implementation decisions that may vary from user to
user.  IDR is well suited to this purpose, as it is entirely generic,
does not involve any implementation decisions, other than the choice
of the partial order, applies to all types of real-valued outcomes
with discrete, continuous or mixed discrete-continuous distributions,
and accommodates general types of covariate spaces.

\section{Isotonic distributional regression}  \label{sec:IDR}

We proceed to introduce the isotonic distributional regression (IDR)
technique.  To this end, we first review basic facts on proper scoring
rules and notions of calibration.  Then we define the IDR solution,
prove existence, uniqueness and universality, and discuss its
computation and asymptotic consistency.  Thereafter, we turn from
estimation to prediction and describe how IDR can be used in
out-of-sample forecasting.  Throughout, we identify a Borel
probability measure on the real line $\real$ with its cumulative
distribution function (CDF), and we denote the extended real line by
$\bar{\real} = [-\infty,\infty]$.

\subsection{Preliminaries}  \label{subsec:preliminaries} 

Following \citet{Gneiting2007a}, we argue that distributional
regression techniques should be compared and evaluated using proper
scoring rules.  A {\em proper scoring rule}\/ is a function $\myS :
\cP \times \real \to \bar{\real}$, where $\cP$ is a
suitable class of probability measures on $\real$, such that
$\myS(F,\cdot)$ is measurable for any $F \in \cP$, the integral $\int
\myS(G,y) \diff F(y)$ exists, and
\[
\int \myS(F,y) \diff F(y) \leq \int \myS(G,y) \diff F(y)
\]
for all $F,G \in \cP$.  A key example is the {\em continuous ranked
  probability score}\/ (CRPS), which is defined for all Borel
probability measures, and given as
\[
{\rm CRPS}(F,y) = \int_{\real} \left( F(z) - \one \{ y \le z \} \right)^2 \diff z. 
\]
Introduced by \citet{Matheson1976}, the CRPS has become popular across
application areas and methodological communities, both for the
purposes of evaluating predictive performance and as a loss function
in estimation; see, e,g., \citet{Hersbach2000}, \citet{Gneiting2005b},
\citet{Hothorn2014}, \citet{Pappenberger2015}, \citet{Rasp2018} and
\citet{Gasthaus2019}.  The CRPS is reported in the same unit as the
response variable, and it reduces to the absolute error, $|x-y|$, if
$F$ is the point or Dirac measure in $x \in \real$.

Results in \citet{Laio2007}, \citet{Ehm2016} and
\citet{Bouallegue2018} imply that the CRPS can be represented
equivalently as
\begin{align}  
{\rm CRPS}(F,y) 
& = 2 \int_{(0,1)} \textrm{QS}_\alpha(F,y) \diff \alpha \label{eq:CRPS2} \\ 
& = 2 \int_{(0,1)} \int_{\real} \myS^Q_{\alpha, \hsp \theta}(F,y) \diff \theta \diff \alpha \label{eq:CRPS3} \\ 
& = \int_{\real} \int_{(0,1)} \myS^P_{z, \hsp c}(F,y) \diff c \diff z, \label{eq:CRPS4}
\end{align}
where the mixture representation \eqref{eq:CRPS2} is in terms of the
asymmetric piecewise linear or pinball loss,
\begin{equation}  \label{eq:QS}
\textrm{QS}_\alpha(F,y) = 
\begin{cases} 
(1 - \alpha) \, (F^{-1}(\alpha) - y), & y \leq F^{-1}(\alpha), \\ 
\alpha \, ( \hsp y - F^{-1}(\alpha)), & y \geq F^{-1}(\alpha), \\ 
\end{cases} 
\end{equation}
which is customarily thought of as a quantile loss function, but can
be identified with a proper scoring rule \citep[Theorem
  3]{Gneiting2011a}.  The representations \eqref{eq:CRPS3} and
\eqref{eq:CRPS4} express the CRPS in terms of the {\em elementary}\/
or {\em extremal scoring functions}\/ for the $\alpha$-quantile
functional, namely,
\begin{equation}  \label{eq:S_Q}
\myS^Q_{\alpha, \hsp \theta}(F,y) = \begin{cases} 
1 - \alpha, & y \le \theta < F^{-1}(\alpha), \\
\alpha, & F^{-1}(\alpha) \le \theta < y, \\
0, & \textrm{otherwise}, \\ 
\end{cases} 
\end{equation}
where $\theta \in \real$; and for probability assessments of the
binary outcome $\one \{ y \leq z \}$ at the threshold value $z \in
\real$, namely
\begin{equation}  \label{eq:S_P}
\myS^P_{z, \hsp c}(F,y) = \begin{cases} 
1 - c, & F(z) < c, \ y \leq z, \\
c, & F(z) \geq c, \ y > z, \\
0, & \textrm{otherwise}, \\ 
\end{cases} 
\end{equation}
where $c \in (0,1)$.  For background information on elementary or
extremal scoring functions and related concepts see \citet{Ehm2016}.

Predictive distributions ought to be calibrated \citep{Dawid1984,
  Diebold1998, Gneiting2007b}, in the broad sense that they should be
statistically compatible with the responses, and various notions of
calibration have been proposed and studied.  In the spirit of
\citet{Gneiting2013}, we consider the joint distribution $\myP$ of the
response $Y$ and the distributional regression $F_X$.  The most widely
used criterion is {\em probabilistic calibration}, which requires that
the {\em probability integral transform}\/ (PIT), namely, the random
variable
\begin{equation}  \label{eq:PIT} 
Z = F_X(Y-) + V \left( F_X(Y) - F_X(Y-) \right) \! ,  
\end{equation} 
be standard uniform, where $F_X(Y-) = \lim_{y \uparrow Y} F_X(y)$ and
$V$ is a standard uniform variable that is independent of $F_X$ and
$Y$.  If $F_X$ is continuous the PIT is simply $Z = F_X(Y)$.  Here we
introduce the novel notion of {\em threshold calibration}, requiring
that
\begin{equation}  \label{eq:threshold_calibration} 
\myP(Y \leq y \hsp \hsp| \hsp F_X(y)) = F_X(y)  
\end{equation} 
almost surely for $y \in \real$, which implies {\em marginal
  calibration}, defined as $\myP(Y \leq y) = \myE(F_X(y))$ for $y \in
\real$.  If $F_X = \cL(Y \hsp | \hsp X)$ then it is calibrated in any
of the above senses \citep[Theorem 2.8]{Gneiting2013}.
 
\subsection{Existence, uniqueness and universality}  \label{subsec:theory} 

A partial order relation $\preceq$ on a set $\cX$ has the same
properties as a total order, namely reflexivity, antisymmetry and
transitivity, except that the elements need not be comparable, i.e.,
there might be elements $x \in \cX$ and $x' \in \cX$ such that neither
$x \preceq x'$ nor $x' \preceq x$ holds.  A key example is the
componentwise order on $\real^n$.

For a positive integer $n$ and a partially ordered set $\cX$, we
define the classes
\begin{align*}
\cX^n_{\uparrow} &= \{ \bx =(x_1,\dots,x_n) \in \cX^n: x_1 \preceq \cdots \preceq x_n \}, \\
\cX^n_{\downarrow} &= \{ \bx=(x_1,\dots,x_n) \in \cX^n: x_1 \succeq \cdots \succeq x_n \}
\end{align*}
of the increasingly and decreasingly (totally) ordered tuples in
$\cX$, respectively.  Similarly, given a further partially ordered set
$\cQ$ and a vector $\bx=(x_1,\dots,x_n) \in \cX^n$, the classes
\begin{align*}
\cQ^n_{\uparrow, \hsp \bx} &= \{ \bq =(q_1,\dots,q_n)\in \cQ^n : q_i \preceq q_j \textrm{ if } x_i \preceq x_j \},\\
\cQ^n_{\downarrow, \hsp \bx} &= \{ \bq =(q_1,\dots,q_n)\in \cQ^n : q_i \succeq q_j \textrm{ if } x_i \preceq x_j \} 
\end{align*}
comprise the increasingly and decreasingly (partially) ordered tuples
in $\cQ$, with the order induced by the tuple $\bx$ and the partial
order $\preceq$ on $\cX$.

Let $I \subseteq \real$ be an interval, and let $\myS$ be a proper
scoring rule with respect to a class $\cP$ of probability
distributions on $I$ that contains all distributions with finite
support.  Given training data in the form of a covariate vector $\bx =
(x_1, \ldots, x_n) \in \cX^n$ and response vector $\by = (y_1, \ldots,
y_n) \in I^n$, we may interpret any mapping from $\bx \in \cX^n$ to
$\cP^n$ as a distributional regression function.  Throughout, we equip
$\cP$ with the usual stochastic order. 

\begin{defn}[$\myS$-based regression]  \label{def:regression} 
An element $\bhF = (\hat{F}_1, \ldots, \hat{F}_n) \in \cP^n$ is an
$\myS$-\textit{based isotonic regression}\/ of $\by \in I^n$ on $\bx
\in \cX^n$, if it is a minimizer of the empirical loss
\[
\ell_{\myS} (\bF) = \frac{1}{n} \sum_{i=1}^n \myS(F_i, y_i)
\]
over all $\bF = (F_1, \dots, F_n)$ in $\cP^n_{\uparrow, \hsp \bx}$.
\end{defn}

In plain words, an $\myS$-based isotonic regression achieves the best
fit in terms of the scoring rule $\myS$, subject to the conditional
CDFs $\hat{F}_1, \ldots, \hat{F}_n$ satisfying partial order
constraints induced by the covariate values $x_1, \ldots, x_n$.  The
definition and the subsequent results can be extended to losses of the
form $\ell_{\myS} (\bF) = \sum_{i=1}^n w_i \hsp \hsp \myS(F_i, y_i)$
with rational, strictly positive weights $w_1, \ldots, w_n$.  The
adaptations are straightforward and left to the reader.

Furthermore, the definition of an $\myS$-based isotonic regression as
a minimizer of $\ell_{\myS}$ continues to apply when $\mathcal{X}$ is
equipped with a pre- or quasiorder $\preceq$ instead of a partial
order.  Preorders are not necessarily antisymmetric, and so there
might be elements $x, x'$ such that $x \preceq x'$ and $x' \preceq x$
but $x' \neq x$.  In this setting, we define $x$ and $x'$ to be
equivalent if $x \preceq x'$ and $x' \preceq x$, and set $[x]
\preceq_p [x']$ if representatives $u, u'$ of the equivalence classes
$[x], [x']$ satisfy $u \preceq u'$.  Then the binary relation
$\preceq_p$ defines a partial order on the set of equivalence classes,
and the $\myS$-based isotonic regression with the new covariates and
the partial order $\preceq_p$ coincides with the original $\myS$-based
isotonic regression.

In Appendix \ref{app:theory} we prove the following result. 

\begin{thm}[existence and uniqueness]  \label{thm:existence}
There exists a unique {\rm CRPS}-based isotonic regression\/ $\bhF \in
\cP^n$ of\/ $\by$ on\/ $\bx$.
\end{thm}

We refer to this unique $\bhF$ as the {\em isotonic distributional
  regression} (IDR) of $\by$ on $\bx$.  In the particular case of a
total order on the covariate space, and assuming that $x_1 < \cdots <
x_n$, for each $z \in I$ the solution $\bhF(z) = (\hat{F}_1(z),
\ldots, \hat{F}_n(z))$ is given by
\begin{equation}  \label{eq:solution}
\hat{F}_i(z) = \min_{k = 1, \ldots, i} \max_{j = k, \ldots, n} \,
\frac{1}{j - k + 1} \sum_{l=k}^j \one \{ y_l \leq z \}
\end{equation}
for $i = 1, \ldots, n$; see eqs.~(1.9)--(1.13) of \citet{Barlow1972}.
A similar max--min formula applies under partial orders
\citep{Robertson1980, Jordan2021}, and it follows that $\hat{F}_i$ is
piecewise constant with any points of discontinuity at $y_1, \ldots,
y_n$.

At first sight, the specific choice of the CRPS as a loss function may
seem arbitrary.  However, the subsequent result, which we prove in
Appendix \ref{app:theory}, reveals that IDR is
simultaneously optimal with respect to broad classes of proper scoring
rules that include all relevant choices in the extant literature.  The
popular logarithmic score allows for the comparison of absolutely
continuous distributions with respect to a fixed dominating measure
only and thus is not applicable here.  Statements concerning
calibration are with respect to the empirical distribution of the
training data $(x_1,y_1), \ldots, (x_n,y_n)$.

\begin{thm}[universality]  \label{thm:universality}
The {\rm IDR} solution\/ $\bhF$ of\/ $\by$ on\/ $\bx$ is threshold
calibrated and has the following properties.
\begin{itemize}
\item[i)] The\/ {\rm IDR} solution\/ $\bhF$ is an {\rm
  $\myS$}-based isotonic\/ regression of\/ $\by$ on\/ $\bx$ under any scoring rule of
  the form
  \begin{equation}  \label{eq:CRPS_Q}
  \myS(F,y) = \int_{(0,1) \times \real} \myS^Q_{\alpha, \hsp \theta}(F,y) \diff H(\alpha,\theta)
  \end{equation}
  or 
  \begin{equation}  \label{eq:CRPS_P}
  \myS(F,y) = \int_{\real \times (0,1)} \myS^P_{z, \hsp c}(F,y) \diff M(z,c),
  \end{equation}
  where\/ $\myS^Q_{\alpha, \hsp \theta}$ is the elementary quantile
  scoring function\/ \eqref{eq:S_Q}, $\myS^P_{z, \hsp c}$ is the
  elementary probability scoring rule \eqref{eq:S_P}, and\/ $H$ and\/
  $M$ are locally finite Borel measures on\/ $(0,1) \times \real$
  and\/ $\real \times (0,1)$, respectively.
\item[ii)] For every\/ $\alpha \in (0,1)$ it holds that\/
  $\bhF^{-1}(\alpha) = (\hat{F}_1^{-1}(\alpha), \ldots,
  \hat{F}_n^{-1}(\alpha))$ is a minimizer of
  \begin{equation}  \label{eq:iso-quant}
  \frac{1}{n} \sum_{i=1}^n \mys_\alpha(\theta_i,y_i)
  \end{equation}
  over all\/ $\pmb{\theta} = (\theta_1, \ldots, \theta_n) \in
  I^n_{\uparrow, \hsp \bx}$, under any function\/ $\mys_\alpha : I
  \times I \to \bar{\real}$ which is left-continuous in both arguments
  and such that\/ $\myS(F,y) = \mys_\alpha(F^{-1}(\alpha),y)$ is a
  proper scoring rule on\/ $\cP$.
\item[iii)] For every threshold value\/ $z \in I$, it is true that
  $\bhF(z) = (\hat{F}_1(z), \ldots, \hat{F}_n(z))$ is a minimizer of
  \begin{equation}   \label{eq:z_optimal}
  \frac{1}{n} \sum_{i=1}^n \mys(\eta_i, \one \{ y_i \leq z \})
  \end{equation}
  over all ordered tuples\/ $\pmb{\eta} = (\eta_1, \ldots, \eta_n) \in
  [0,1]^n_{\downarrow, \hsp \bx}$, under any function\/ $\mys : [0,1]
  \times \{ 0, 1 \} \to \bar{\real}$ that is a proper scoring rule for
  binary events, which is left-continuous in its first argument,
  satisfies\/ $\mys(0,y) = \lim_{p \to 0}\mys(p,y)$, and is
  real-valued, except possibly\/ $\mys(0,1) = - \infty$ or\/
  $\mys(1,0) = -\infty$.
  \end{itemize}
\end{thm}

The quantile weighted and threshold weighted versions of the CRPS
studied by \citet{Gneiting2011b} arise from \eqref{eq:CRPS_Q} and
\eqref{eq:CRPS_P} with $H = G_0 \otimes \lambda$ and $M = \lambda
\otimes G_1$, where $\lambda$ denotes the Lebesgue measure, and $G_0$
and $G_1$ are $\sigma$-finite Borel measures on $(0,1)$ and $\real$,
respectively.  If $G_0$ and $G_1$ are Lebesgue measures, we recover
the mixture representations \eqref{eq:CRPS3} and \eqref{eq:CRPS4} of
the CRPS.  By results of \citet{Ehm2016}, if $H$ is concentrated on
$\{ \alpha \} \times \real$ and $M$ is concentrated on $\{ z \} \times
(0,1)$, these representations cover essentially all proper scoring
rules that depend on the predictive distribution $F$ via
$F^{-1}(\alpha)$ or $F(z)$ only, yielding universal optimality in
statements in parts ii) and iii) of Theorem \ref{thm:universality}.

In particular, as a special case of \eqref{eq:iso-quant}, the IDR
solution is a minimizer of the quantile loss under the asymmetric
piecewise linear or pinball function \eqref{eq:QS} that lies at the
heart of quantile regression \citep{Koenker2005}.  Consequently, as
the mixture representation \eqref{eq:CRPS2} of the CRPS may suggest,
IDR nests classical nonparametric isotonic quantile regression as
introduced and studied by \citet{Robertson1975} and
\citet{Casady1976}.  In other words, part ii) of Theorem
\ref{thm:universality} demonstrates that, if we (hypothetically)
perform nonparametric isotonic quantile regression at every level
$\alpha \in (0,1)$ and piece the conditional quantile functions
together, we recover the IDR solution.  However, the IDR solution is
readily computable (Section \ref{subsec:computing}), without invoking
approximations or truncations, unlike brute force approaches to
simultaneous quantile regressions.  Loss functions of the form
\eqref{eq:iso-quant} also include the interval score
\citep[eq.~(43)]{Winkler1972, Gneiting2007a}, which constitutes the
most used proper performance measure for interval forecasts.

In the special case of a binary response variable, we see from iii)
and \eqref{eq:z_optimal} that the IDR solution is an $\myS$-based
isotonic regression under just any applicable proper scoring rule
$\myS$.  Furthermore, threshold calibration is the strongest possible
notion of calibration in this setting \citep[Theorem
  2.11]{Gneiting2013}, so the IDR solution is universal in every
regard.  In the further special case of a total order on the covariate
space, the IDR and pool adjacent violators (PAV) algorithm solutions
coincide, and the statement in iii) is essentially equivalent to
Theorem 1.12 of \citet{Barlow1972}.  In particular, the IDR or PAV
solution yields both the nonparametric maximum likelihood estimate and
the nonparametric least squares estimate under the constraint of
isotonicity.  The latter suggests a computational implementation via
quadratic programming, to which we tend now.

\subsection{Computational aspects}  \label{subsec:computing} 

The key observation towards a computational implementation is the
aforementioned special case of \eqref{eq:z_optimal}, according to
which the IDR solution $\bhF \in \cP^n$ of $\by \in \real^n$ on $\bx
\in \cX^n$ satisfies
\begin{equation}  \label{eq:QP} 
\bhF(z) = \arg \, \min_{\eta \in [0, 1]^n_{\downarrow, \hsp \bx}} 
\sum_{i=1}^n \left( \eta_i - \one \{ y_i \leq z \} \right)^2
\end{equation}
at every threshold value $z \in \real$.  In this light, the
computation of the IDR CDF at any fixed threshold reduces to a
quadratic programming problem.  The above target function is constant
in between the unique values of $y_1, \ldots, y_n$, say $\tilde{y}_1 <
\cdots < \tilde{y}_m$, and so it suffices to estimate the CDFs at
these points only.  {\green In contrast, exact implementations based
  on quantiles would need to consider all levels of the form $i/j$
  with integers $1 \leq i < j \leq n$, which is computationally
  prohibitive.  In the threshold-based approach, the} overall cost
depends on the quadratic programming solver applied, and the
computation becomes much faster if recursive relations between
consecutive conditional CDFs $\bhF(\tilde{y}_k)$ and
$\bhF(\tilde{y}_{k-1})$ are taken advantage of.  In the case of a
total order, \citet{Henzi2020} describe a recursive adaptation of the
PAV algorithm to IDR that considerably reduces the computation time
as compared to a naive implementation which does not take into account
recursive relations.
Under general partial orders, active set methods for solutions to the
quadratic programming problem \eqref{eq:QP} have been discussed by
\citet{deLeeuw2009}.  In our implementation, we use the powerful
quadratic programming solver OSQP \citep{Stellato2020} as supplied by
the package {\tt osqp} in the statistical programming environment
\textsf{R} \citep{Stellato2018, R}, which can be warm-started
efficiently by taking $\bhF(\tilde{y}_{k-1})$ as a starting point for
the computation of $\bhF(\tilde{y}_{k})$.

Clearly, a challenge in the computational implementation of IDR with
general partial orders is that the number of variables in the
quadratic programming problem \eqref{eq:QP} grows at a rate of
$\cO(n)$.  As a remedy, we propose subsample aggregation, much in the
spirit of random forests that rely on bootstrap aggregated (bagged)
classification and regression trees \citep{Breiman1996, Breiman2001}.
It was observed early on that random forests generate conditional
predictive distributions \citep{Hothorn2004, Meinshausen2006}, and
recent applications include the statistical postprocessing of ensemble
weather forecasts \citep{Taillardat2016, Schlosser2019,
  Taillardat2019}.  \citet{Buhlmann2002} and \citet{Buja2006} argue
forcefully that subsample aggregation (subagging) tends to be equally
effective as bagging, but at considerably lower computational cost.

\begin{figure}[t]
\centering
\includegraphics[width = \textwidth]{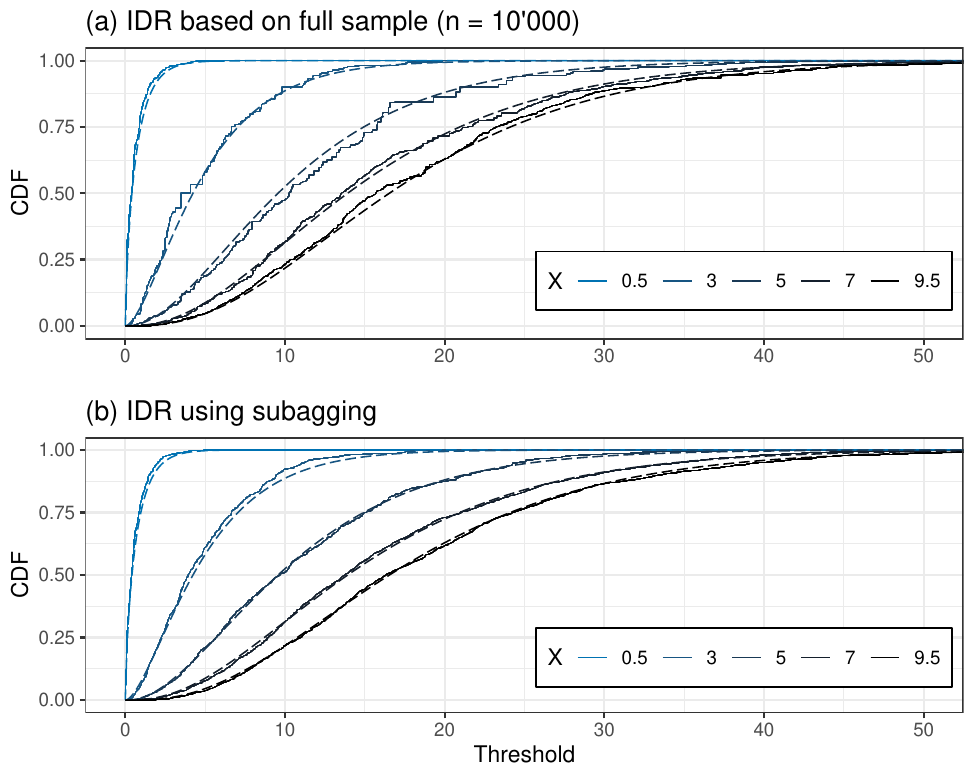}
\caption{Simulation example for a sample of size $n = 10\,000$ from
  the distribution in \eqref{eq:sim}. The true conditional CDFs
  (smooth dashed graphs) are compared to IDR estimates (step
  functions) based on (a) the full training sample of size $n =
  10\,000$ and (b) linear aggregation of IDR estimates on 100
  subsamples of size $1\,000$ each.  \label{fig:subagging}}
\end{figure}

In view of the superlinear computational costs of IDR, smart uses of
subsample aggregation yield major efficiency gains, taking into
account that the estimation on different subsamples can be performed
in parallel.  Isotonicity is preserved under linear aggregation, and
the aggregated conditional CDFs can be inverted to generate isotonic
conditional quantile functions, with the further benefit of smoother
estimates in continuous settings.  A detailed investigation of optimal
subsample aggregation for IDR is a topic for future research. For
illustration, Figure \ref{fig:subagging} returns to the simulation
example in Figure \ref{fig:sim}, but now with a much larger training
sample of size $n = 10\,000$ from the distribution in \eqref{eq:sim}.
Linear aggregation based on 100 subsamples (drawn without replacement)
of size $n = 1\,000$ each is superior to the brute force approach on
the full training set in terms of estimation accuracy.  The
computation on the full dataset for this simulation example takes 11.7
seconds for the naive implementation, but only 1.1 seconds for the
sequential algorithm of \citet{Henzi2020}.  Subagging gives
computation times of 9.9 and 2.5 seconds, respectively, or 1.8 and 0.5
seconds when parallelized over eight cores.\footnote{With Intel(R)
  Xeon(R) E5-2630 v4 2.20GHz CPUs, in \textsf{R} \citep{R}, using the
  {\tt doParallel} package for parallelization.  Times reported are
  averages over 100 replicates.}

\subsection{{Consistency}} \label{subsec:consistency}

We proceed to prove uniform consistency of the IDR estimator.  While
strong consistency of nonparametric isotonic quantile regression for
single quantiles was proved decades ago \citep{Robertson1975,
  Casady1976}, uniform consistency and rates of convergence for the
IDR estimator have been established only recently, and exclusively in
the case of a total order, see \citet[Theorem 1]{ElBarmi2005} and
\citet[Theorem 3.3]{Moesching2020}.

For $x \in \cX$ and $y \in \real$, let $\hat{F}_x(y)$ denote the IDR
estimate based on fixed or random pairs $(X_1, Y_1), \ldots, (X_n,
Y_n)$.  As introduced thus far, the IDR solution $\bhF = (\hat{F}_1,
\ldots, \hat{F}_n)$ is defined at the covariate values $X_1, \ldots,
X_n \in \cX$ only.  For general $x \in \cX$, we merely assume that
$\hat{F}_x(y)$ is some value in between the bounds given by
\begin{equation}  \label{eq:estbounds}
\max_{i \in s(x)} \hat{F}_i(y) \leq \hat{F}_x(y) \leq \min_{i \in p(x)} \hat{F}_i(y).
\end{equation}
Here, we define the sets of the indices of {\em direct predecessors}\/
and {\em direct successors}\/ of $x\in \cX$ among the covariate values
as
\begin{equation}\label{eq:predec}
p(x) = \{ i \in \{ 1, \ldots, n \} : 
          X_i \preceq X_j \preceq x \implies X_j = X_i, \, j = 1, \ldots, n \}
\end{equation}
and 
\begin{equation}\label{eq:suc}
s(x) = \{ i \in \{ 1, \ldots, n \} : 
          x \preceq X_j \preceq X_i \implies X_j = X_i, \, j = 1, \ldots, n \},
\end{equation}
respectively.  

In Appendix \ref{app:consistency} we establish
the following consistency theorem, which covers key examples of
partial orders and is based on strictly weaker assumptions than the
results of \citet{Moesching2020}.  However, in contrast to their work,
we do not provide rates of convergence.  The choice $\cX = [0,1]^d$
for the covariate space merely serves to simplify the presentation: As
IDR is invariant under strictly isotonic transformations, any
covariate vector $X = (X_1, \ldots, X_d) \in \real^d$ can be
transformed to have support in $[0,1]^d$, and the componentwise
partial order can be replaced by any weaker preorder.  A key
assumption uses the concept of an \emph{antichain}\/ in a partially
ordered set $(\cS, \preceq)$, which is a subset $A \subseteq \cS$ that
does not admit comparisons, in the sense that $u \preceq v$ for $u, v
\in A$ implies $u = v$.  As we discuss subsequently, results of
\citet{Brightwell1992} imply that the respective distributional
condition is mild.

\begin{thm}[uniform consistency] \label{thm:consistency} 
Let\/ $\cX = [0,1]^d$ be endowed with the componentwise partial order
and the norm\/ $\| u \| = \max_{i = 1, \ldots, d} |u_i|$.  Let
further\/ $(X_{ni}, Y_{ni}) \in [0,1]^d \times \real$, $n \in
\mathbb{N}$, $i = 1, \ldots,n$, be a triangular array such that
$(X_{n1},Y_{n1}),\dots,$ $(X_{nn},Y_{nn})$ are independent and
identically distributed random vectors for each $n \in \mathbb{N}$,
and let\/ $S_n = \{ X_{n1}, \ldots, X_{nn} \}$.  Assume that
\begin{itemize}
\item[(i)] for all non-degenerate rectangles\/ $J \subseteq \cX$,
  there exists a constant $c_J > 0$ such that
  \[
 \# (S_n \cap J) \ge n c_J 
  \]
  with asymptotic probability one, i.e., if\/ $A_n$ denotes the event
  that\/ $\# (S_n \cap J) \ge n c_J$, then\/ $\myP(A_n) \to 1$ as\/ $n
  \to \infty$;
\item[(ii)] for some $\gamma \in (0,1)$,
  \[
 \max \{ \#A : \, A \subset S_n \mathrm{ \ is \ antichain} \} \le n^\gamma
  \]
  with asymptotic probability one.
\end{itemize}
\noindent Assume further that the true conditional CDFs\/ $F_x(y) =
\myP( Y_{ni} \leq y \mid X_{ni} = x)$ satisfy
\begin{itemize}
\item[(iii)] $F_x(y)$ is decreasing in\/ $x$ for all\/ $y \in \real$;
\item[(iv)] for every\/ $\eta > 0$, there exists\/ $r > 0$ such that
  \[
  \sup \{ |F_x(y) - F_{x'}(y)| : \, x, x' \in [0,1]^d, \| x - x' \|
  \leq r, \, y \in \real \} < \eta.
  \]
\end{itemize}
Then for every\/ $\epsilon > 0$ and\/ $\delta > 0$,
\begin{equation} \label{eq:consistency}
\lim_{n \to \infty} \myP \left( \sup_{x \in [\delta, 1-\delta]^d, \, y
  \in \real} |\hat{F}_x(y) - F_x(y)| \geq \epsilon \right) = 0.
\end{equation}
\end{thm}

Assumption (i) requires that the covariates are sufficiently dense in $\cX$, as
is satisfied under strictly positive Lebesgue densities on $\cX$.  In order to
derive rates of convergence, the size of the rectangles $J$ in (i)
would need to decrease with $n$, as in condition (A.2) of
\citet{Moesching2020}; we leave this type of extension as a direction
for future work.  Assumption (iii) is the basic model assumption of
IDR, while assumption (iv) requires uniform continuity of the
conditional distributions, which is weaker than H\"older continuity in
condition (A.1) of \citet{Moesching2020}.

Assumption (ii), which is always satisfied in the case of a total
order, calls for a more detailed discussion.  In words, the maximal
number of mutually incomparable elements needs to grow at a rate
slower than $n^{\gamma}$.  Evidently, the easier elements can be
ordered, the smaller the maximal antichain.  Consequently, Theorem
\ref{thm:consistency} continues to hold under the empirical stochastic
order and the empirical increasing convex order on the covariates
introduced in Section \ref{subsec:icx}, and indeed under any preorder
that is weaker than the componentwise order.  The key to understanding
the distributional implications of (ii) is Corollary 2 in
\citet{Brightwell1992}, which states that for a sequence of
independent random vectors from a uniform population on $[0,1]^d$ the
size of a maximal antichain grows at a rate of $n^{1-1/d}$; see also
the remark following the proof of Theorem \ref{thm:consistency} in
Appendix \ref{app:consistency}.

As comparability under the componentwise order is preserved under
monotonic transformations, \emph{any}\/ covariate vector $X \in
\real^d$ that can be obtained as a monotone transformation of a
uniform random vector of arbitrary dimension guarantees (ii).  This
includes, e.g., all Gaussian random vectors with nonnegative
correlation coefficients.  In this light, assumption (ii) is rather
weak, and well in line with the intuition that for multivariate
isotonic (distributional) regression to work well, there ought be at
least minor positive dependence between the covariates.  In the
context of our case study in Section \ref{sec:case_study}, high
positive correlations between the covariates are the rule, as
exemplified by Table 3 in \citet{Raftery2005}.

\subsection{Prediction}  \label{subsec:prediction}

As noted, the IDR solution $\bhF = (\hat{F}_1, \ldots, \hat{F}_n) \in
\cP^n_{\uparrow, \hsp \bx}$ is defined at the covariate values $x_1,
\ldots, x_n \in \cX$ only.  Generally, if a (not necessarily optimal)
distributional regression $\bF = (F_1, \ldots, F_n) \in
\cP^n_{\uparrow, \hsp \bx}$ is available, a key task in practice is to
make a prediction at a new covariate value $x \in \cX$ where $x
\not\in \{ x_1, \ldots, x_n \}$.  We denote the respective predictive
CDF by $F$.

In the specific case $\cX = \real$ of a single real-valued covariate
there is a simple way of doing this, as frequently implemented in
concert with the PAV algorithm.  For simplicity we suppose that $x_1 <
\cdots < x_n$.  If $x < x_1$ we may let $F = F_1$; if $x \in (x_i,
x_{i+1})$ for some $i \in \{ 1, \ldots, n - 1 \}$ we may interpolate
linearly, so that
\[
F(z) = \frac{x-x_i}{x_{i+1}-x_i} F_i(z) + \frac{x_{i+1}-x}{x_{i+1}-x_i} F_{i+1}(z)
\]
for $z \in \real$, and if $x > x_n$ we may set $F = F_n$.  However,
approaches that are based on interpolation do not extend to a generic
covariate space, which may or may not be equipped with a metric.

In contrast, the method we describe now, which generalizes a proposal
by \citet{Wilbur2005}, solely uses information supplied by the partial
order $\preceq$ on the covariate space $\cX$.  For a general covariate
value $x \in \cX$, the sets of the indices of direct predecessors and
direct successors among the covariate values $x_1, \ldots, x_n$ in the
training data is defined as at \eqref{eq:predec} and \eqref{eq:suc},
respectively with $X_1, \ldots, X_n$ replaced by $x_1, \ldots, x_n$.
If the covariate space $\cX$ is totally ordered, these sets contain at
most one element.  If the order is partial but not total, $p(x)$ and
$s(x)$ may, and frequently do, contain more than one element.
Assuming that $p(x)$ and $s(x)$ are non-empty, any predictive CDF $F$
that is consistent with $\bF$ must satisfy
\begin{equation}  \label{eq:predbounds}
\max_{i \in s(x)} F_i(z) \leq F(z) \leq \min_{i \in p(x)} F_i(z)
\end{equation}
at all threshold values $z \in \real$.  We now let $F$ be the
pointwise arithmetic average of these bounds, i.e.,
\begin{equation}  \label{eq:prediction} 
F(z) = \frac{1}{2} \left( \, \max_{i \in s(x)} F_i(z) + \min_{i \in p(x)} F_i(z) \right)
\end{equation}
for $z \in \real$.  If $s(x)$ is empty while $p(x)$ is non-empty, or
vice-versa, a natural choice, which we employ hereinafter, is to let
$F$ equal the available bound given by the non-empty set.  If $x$
is not comparable to any of $x_1, \ldots, x_n$ the training data lack
information about the conditional distribution at $x$, and a natural
approach, which we adopt and implement, is to set $F$ equal to the
empirical distribution of the response values $y_1, \ldots, y_n$.

The difference between the bounds (if any) in \eqref{eq:predbounds}
might be a useful measure of estimation uncertainty and could be
explored as a promising avenue towards the quantification of ambiguity
and generation of second-order probabilities \citep{Ellsberg1961,
  Seo2009}.  In the context of ensemble weather forecasts, the
assessment of ambiguity has been pioneered by \citet{Allen2012}.
Interesting links arise when the envelope in \eqref{eq:predbounds} is
interpreted in the spirit of randomized predictive systems and
conformal estimates as studied by \cite{Vovk2019}; compare, e.g.,
their Figure 5 with our Figure \ref{fig:example_CDFs}b below.

\section{Partial orders}  \label{sec:orders}

The choice of a sufficiently informative partial order on the
covariate space is critical to any successful application of IDR.  In
the extreme case of distinct, totally ordered covariate values $x_1,
\ldots, x_n \in \cX$ and a perfect monotonic relationship to the
response values $y_1, \ldots, y_n$, the IDR distribution associated
with $x_i$ is simply the point measure in $y_i$, for $i = 1, \ldots,
n$.  The same happens in the other extreme, when there are no order
relations at all.  Hence, the partial order serves to regularize the
IDR solution. 

Thus far, we have simply assumed that the covariate space $\cX$ is
equipped with a partial order $\preceq$, without specifying how the
order might be defined.  If $\cX \subseteq \real^d$, the usual
componentwise order will be suitable in many applications, and we
investigate it in Section \ref{subsec:cw}.  For covariates that are
ordinal and admit a ranking in terms of importance, a
lexicographic order may be suitable.

If groups of covariates are exchangeable, as in our case study on
quantitative precipitation forecasts, other types of order relations
need to be considered.  In Sections \ref{subsec:st} and
\ref{subsec:icx} we study relations that are tailored to this setting,
namely, the empirical stochastic order and empirical increasing convex
order.  Proofs are deferred to Appendix \ref{app:orders}.

\subsection{Componentwise order}  \label{subsec:cw}

Let $x = (x_1, \ldots, x_d)$ and $x' = (x'_1, \ldots,
x'_d)$ denote elements of the covariate space $\real^d$.  The most
commonly used partial order in multivariate isotonic regression is the
\emph{componentwise order}\/ defined by
\[
x \preceq x' \; \iff \; x_i \leq x'_i \, \textrm{ for } \, i = 1, \ldots, d.
\]
This order becomes weaker as the dimension $d$ of the covariate space
increases: If $\tilde{x} = (x_1, \ldots, x_d, x_{d+1})$ and $\tilde{x}' =
(x'_1, \ldots, x'_d, x'_{d+1})$ then $x \preceq x'$ is a
necessary condition for $\tilde{x} \preceq \tilde{x}'$.  The following result is
an immediate consequence of this observation and the structure of the
optimization problem in Definition \ref{def:regression}.

\begin{prop}  \label{prop:cw}
Let\/ $\bx = (x_1, \ldots, x_n)$ and\/ $\bx^* = (x_1^*, \ldots, x_n^*)$
have components\/ $x_i = (x_{i1}, \ldots, x_{id}) \in \real^d$ and\/
$x_i^* = (x_{i1}, \ldots, x_{id}, x_{i,d+1}) \in \real^{d+1}$ for\/ $i
= 1, \ldots, n$, and let\/ $\myS$ be a proper scoring rule.

Then if\/ $\real^d$ and\/ $\real^{d+1}$ are equipped with the
componentwise partial order, and\/ $\bhF$ and\/ $\bhF^*$ denote
$\myS$-based isotonic regressions of\/ $\by$ on\/ $\bx$ and\/ $\bx^*$,
respectively, it is true that
\[
\ell_{\myS}(\bhF^*) \leq \ell_{\myS}(\bhF).
\]
\end{prop}

In simple words, under the componentwise partial order, the inclusion
of further covariates can only improve the in-sample fit.  This
behaviour resembles linear regression, where the addition of
covariates can only improve the (unadjusted) R-square.

\subsection{Empirical stochastic order}  \label{subsec:st}

We now define a relation that is based on stochastic dominance and
invariant under permutation.

\begin{defn} 
Let $x = (x_1, \ldots, x_d)$ and $x' = (x_1', \ldots,
x_d')$ denote elements of $\real^d$.  Then $x$ is smaller than or
equal to $x'$ in {\em empirical stochastic order}, for short $x \st
x'$, if the empirical distribution of $x_1, \ldots, x_d$ is
smaller than the empirical distribution of $x_1', \ldots, x_d'$ in the usual
stochastic order.
\end{defn}

This relation is tailored to groups of exchangeable, real-valued
covariates.  The following results summarizes its properties and
compares to the componentwise order, which we denote by $\preceq$.

\begin{prop}  \label{prop:st}
Let\/ $x = (x_1, \ldots, x_d)$ and\/ $x' = (x_1', \ldots, x_d')$ denote
elements of\/ $\real^d$ with order statistics\/ $x_{(1)} \leq \cdots
\leq x_{(d)}$ and\/ $x'_{(1)} \leq \cdots \leq x'_{(d)}$.
\begin{itemize}
\item[i)] The relation\/ $x \st x'$ is equivalent to\/ $x_{(i)}
  \leq x'_{(i)}$ for $i = 1, \ldots, d$.
\item[ii)] If\/ $x \preceq x'$ then\/ $x \st x'$.
\item[iii)] If\/ $x \st x'$ and\/ $x$ and\/ $x'$ are comparable in the
  componentwise partial order, then\/ $x \preceq x'$.
\item[iv)] If\/ $x \st x'$ and\/ $x' \st x$ then\/ $x$ and\/ $x'$ are
  permutations of each other.  Consequently, the relation\/ $\st$
  defines a partial order on\/ $\real^d_{\uparrow}$.
\end{itemize}
\end{prop}

In a nutshell, the empirical stochastic order is equivalent to the
componentwise order on the sorted elements, and this relation is
weaker than the componentwise order.  However, unlike the
componentwise order, the empirical stochastic order does not
degenerate as further covariates are added.  To the contrary,
empirical distributions of larger numbers of exchangeable variables
become more informative and more easily comparable.

\subsection{Empirical increasing convex order}  \label{subsec:icx} 

In applications, the empirical stochastic order might be too strong,
in the sense that it does not generate sufficiently informative
constraints. In this light, we now define a weaker partial order on
$\real^d_{\uparrow}$, which also is based on a partial order for
probability measures.  Specifically, let $X$ and $X'$ be random
variables with CDFs $F$ and $F'$.  Then $F$ is smaller than $F'$ in
increasing convex order if $\myE(\phi(X)) \leq \myE(\phi(X'))$ for all
increasing convex functions $\phi$ such that the expectations exist
\citep[Section 4.A.1]{Shaked2007}.

\begin{defn} 
Let $x = (x_{1}, \ldots, x_{d})$ and $x' = (x_{1}', \ldots,
x_{d}')$ denote elements of $\real^d$.  Then $x$ is smaller than or
equal to $x'$ in {\em empirical increasing convex order}, for short
$x \icx x'$, if the empirical distribution of $x_{1}, \ldots,
x_{d}$ is smaller than the empirical distribution of $x_{1}', \ldots,
x_{d}'$ in increasing convex order.
\end{defn}

This notion provides another meaningful relation for groups of
exchangeable covariates.  The following result summarizes its
properties and relates it to the empirical stochastic order.

\begin{prop}  \label{prop:icx}
Let\/ $x = (x_1, \ldots, x_d)$ and\/ $x' = (x_1', \ldots, x_d')$ denote
elements of\/ $\real^d$ with order statistics\/ $x_{(1)} \leq \cdots
\leq x_{(d)}$ and\/ $x_{(1)}' \leq \cdots \leq x_{(d)}'$.
\begin{itemize}
\item[i)] The relation\/ $x \icx x'$ is equivalent to
\[
\sum_{i = j}^d x_{(i)} \leq \sum_{i = j}^d x_{(i)}' \; \textrm{ for } \; j = 1, \ldots, d.
\]
\item[ii)] If\/ $x \st x'$ then\/ $x \icx x'$.
\item[iii)] If\/ $x \icx x'$ then
  \[
  \frac{1}{d} \sum_{i=1}^d x_i  + \frac{d-1}{2(d+1)} \hsp \hsp g(x) \ \leq \
  \frac{1}{d} \sum_{i=1}^d x_i'  + \frac{d-1}{2(d+1)} \hsp \hsp g(x'), 
  \]
  where $g$ is the Gini mean difference, 
  \begin{equation}  \label{eq:Gini} 
  g(x) = \frac{1}{d(d-1)} \sum_{i, j = 1}^d |x_i - x_j|.
  \end{equation}
\item[iv)] If\/ $x \icx x'$ and\/ $x' \icx x$ then\/ $x$ and\/ $x'$
  are permutations of each other.  Consequently, the relation\/ $\icx$
  defines a partial order on\/ $\real^d_{\uparrow}$.
\end{itemize}
\end{prop}

\begin{figure}
\center \includegraphics[width = \textwidth]{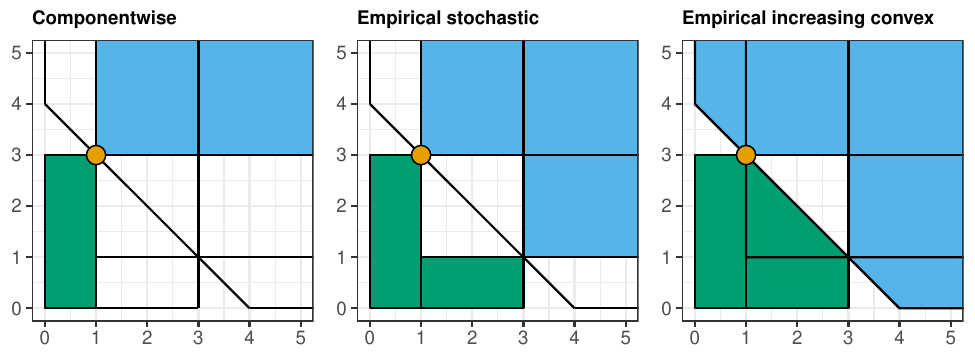}
\caption{Regions of smaller, greater and incomparable elements in the
  positive quadrant of $\real^2$, as compared to the point $(1, 3)$,
  for the (left) componentwise, (middle) empirical stochastic and
  (right) empirical increasing convex order.  Coloured areas below
  (above) of $(1,3)$ correspond to smaller (greater) elements, while
  blank areas contain elements incomparable to $(1,3)$ in the given
  partial order.  \label{fig:part_ord}}
\end{figure}

Figure \ref{fig:part_ord} illustrates the various types of relations
for points in the positive quadrant of $\real^2$.  As reflected by the
nested character of the regions, the componentwise order is stronger
than the empirical stochastic order, which in turn is stronger than
the empirical increasing convex order.  The latter is equivalent to
{\em weak majorization}\/ as studied by \citet{Marshall1979}.  In the
special case of vectors with non-negative entries, their Corollary C.5
implies that $x \in \real^d$ is dominated by $x' \in \real^d$ in
empirical increasing convex order if, and only if, it lies in the
convex hull of the points of the form $(\xi_1 x_{\pi(1)}', \ldots,
\xi_d \hsp \hsp x_{\pi(d)}')$, where $\pi$ is a permutation and $\xi_i
\in \{ 0, 1 \}$ for $i = 1, \ldots, d$.

\section{Simulation study}  \label{sec:simulation_study}

Since we view IDR primarily as a tool for prediction, we compare it to
other distributional regression methods in terms of predictive
performance on continuous and discrete, univariate simulation
examples, as measured by the CRPS.  However, as noted below and
formalized in Appendix \ref{app:CRPS.L2}, the
CRPS links asymptotically to $L_2$ estimation error, so under large
validation samples prediction and estimation are assessed
simultaneously.  A detailed comparative study on mixed
discrete-continuous data with a multivariate covariate vector is given
in the case study in the next section.

Here, our simulation scenarios build on the illustrating example in
the introduction.  Specifically, we draw a covariate $X \sim
\textrm{Unif}(0, 10)$ and then
\begin{align}
& Y_1 \mid X \sim \textrm{Gamma}(\textrm{shape} = \sqrt{X}, \, 
   \textrm{scale} = \min \{ \max \{ X, 1 \}, 6 \}), \label{eq:sim_standard} \\
& Y_2 \mid X = Y_1 \mid X + 10 \cdot \one \{ X \geq 5 \}, \label{eq:sim_discontinuous} \\ 
& Y_3 \mid X = Y_1 \mid X - 2 \cdot \one \{ X \geq 7 \}, \label{eq:sim_noniso} \\ 
& Y_4 \mid X \sim \textrm{Poisson}(\lambda = \min \{ \max \{ X, 1 \}, 6 \}) \}). \label{eq:sim_poisson} 
\end{align}
Under each scenario we generate 500 training sets of size $n = 500$,
$1\,000$, $2\,000$, and $4\,000$ each, fit distributional regression
models, and validate on a test set of size $m = 5\,000$.  For
comparison with IDR, we use a nonparametric kernel (or nearest
neighbor) smoothing technique \citep[NP;][]{Li2008}, semiparametric
quantile regression with monotone rearrangement (SQR;
\citealt{Koenker2005}; \citealt{Chernozhukov2010}), conditional
transformation models \citep[TRAM;][]{Hothorn2014}, and distributional
or quantile random forests (QRF; \citealt{Meinshausen2006};
\citealt{Athey2019}).  These methods have been chosen as they are not
subject to restrictive assumptions on the distribution of the response
variable and have well established and well documented implementations
in the statistical programming environment \textsf{R} \citep{R}.  We
also include the ideal forecast, i.e., the true conditional
distribution of the response given the covariate, in the comparison.

Implementation details for the various methods are given in Table
\ref{tab:sim_description} in Appendix \ref{app:tab_fig}. 
{\green Here we only note that QRF uses the {\tt
    grf} package \citep{Tibshirani2020} with a splitting rule that is
  tailored to quantiles \citep{Athey2019}.}  We see that, unlike IDR,
its competitors rely on manual intervention and tuning.  For example,
QRFs perform poorly under the default value of 5 for the tuning
parameter {\tt min.node.size}, which we have raised to 40.  Further
improvement may arise when tuning parameters, such as honesty fraction
and node size, are judiciously adjusted to the specific scenario and
training sample size at hand.  In contrast, IDR is entirely free of
implementation decisions, except for the subagging variant, \IDRsbg,
where we average predictions based on estimates on 100 subsamples of
size $n/2$ each.

Scenario \eqref{eq:sim_standard} is the same as in the introduction
and illustrated in Figure \ref{fig:sim}.  It has a smooth
covariate--response relationship, and {\green NP, SQR, and even the
  misspecified TRAM technique, which are tailored to this type of
  setting, outperform QRF and IDR.}  However, the assumption of
continuity in the response is crucial, as the results under the
discontinuous scenario \eqref{eq:sim_discontinuous} demonstrate, where
IDR and \IDRsbg \ perform best.  In the non-isotonic scenario
\eqref{eq:sim_noniso} IDR and \IDRsbg \ retain acceptable performance,
even though the key assumption is violated.  Not surprisingly, SQR faces
challenges in the Poisson scenario \eqref{eq:sim_poisson},
where the conditional quantile functions are piecewise constant, and
IDR is outperformed only by TRAM.  Throughout, the simplistic
subagging variant of IDR has slightly lower mean CRPS than the default
variant that is estimated on the full training set, and it would be
interesting to explore the relation to the super-efficiency phenomenon
described by \citet{Banerjee2019}.

\begin{table}
%\centering 
\caption{Mean CRPS in smooth \eqref{eq:sim_standard}, discontinuous
  \eqref{eq:sim_discontinuous}, non-isotonic \eqref{eq:sim_noniso},
  and discrete \eqref{eq:sim_poisson} simulation scenarios with
  training sets of size $n$. \label{tab:sim}}
{\footnotesize
\begin{tabular}{l|cccc|cccc}
\multicolumn{9}{c}{} \\
\toprule & \multicolumn{4}{c|}{Smooth \eqref{eq:sim_standard}}
         & \multicolumn{4}{c}{Discontinuous \eqref{eq:sim_discontinuous}} \\
\midrule
$n$ & 500 & $1\,000$ & $2\,000$ & $4\,000$ & 500 & $1\,000$ & $2\,000$ & $4\,000$ \\
\midrule
NP       & 3.561 & 3.542 & 3.532 & 3.525 & 3.614 & 3.582 & 3.562 & 3.549 \\ 
SQR      & 3.571 & 3.543 & 3.530 & 3.524 & 3.647 & 3.619 & 3.606 & 3.600 \\ 
TRAM     & 3.560 & 3.543 & 3.535 & 3.531 & 3.642 & 3.625 & 3.616 & 3.612 \\ 
QRF	 	 & 3.589 & 3.561 & 3.555 & 3.553 & 3.614 & 3.576 & 3.561 & 3.556 \\ 
IDR      & 3.604 & 3.568 & 3.548 & 3.535 & 3.628 & 3.581 & 3.555 & 3.540 \\ 
\IDRsbg  & 3.595 & 3.561 & 3.543 & 3.532 & 3.620 & 3.577 & 3.551 & 3.537 \\ 
Ideal    & 3.516 & 3.516 & 3.516 & 3.516 & 3.516 & 3.516 & 3.516 & 3.516 \\ 
\midrule
 & \multicolumn{4}{c|}{Non-isotonic \eqref{eq:sim_noniso}}
         & \multicolumn{4}{c}{Discrete \eqref{eq:sim_poisson}} \\
\midrule
$n$ & 500 & $1\,000$ & $2\,000$ & $4\,000$ & 500 & $1\,000$ & $2\,000$ & $4\,000$ \\
\midrule
NP       & 3.564 & 3.544 & 3.534 & 3.527 & 1.136 & 1.131 & 1.128 & 1.126 \\ 
SQR      & 3.574 & 3.546 & 3.533 & 3.527 & 1.129 & 1.121 & 1.116 & 1.114 \\ 
TRAM     & 3.566 & 3.549 & 3.543 & 3.539 & 1.115 & 1.110 & 1.107 & 1.106 \\ 
QRF      & 3.587 & 3.560 & 3.555 & 3.553 & 1.121 & 1.113 & 1.112 & 1.112 \\
IDR      & 3.605 & 3.569 & 3.549 & 3.536 & 1.130 & 1.119 & 1.113 & 1.109 \\ 
\IDRsbg  & 3.597 & 3.564 & 3.545 & 3.534 & 1.128 & 1.118 & 1.112 & 1.109 \\ 
Ideal    & 3.516 & 3.516 & 3.516 & 3.516 & 1.104 & 1.104 & 1.104 & 1.104 \\ 
\bottomrule
\end{tabular}
}
\end{table}

These results lend support to our belief that IDR can serve as a
universal benchmark in probabilistic forecasting and distributional
regression problems.  For sufficiently large training samples, IDR
offers competitive performance under any type of type of linearly
ordered outcome, without reliance on tuning parameters or other
implementation choices, except when subsampling is employed.

\section{Case study: Probabilistic quantitative precipitation forecasts}  \label{sec:case_study}

The past decades have witnessed tremendous progress in the science and
practice of weather prediction \citep{Bauer2015}.  Arguably, the most
radical innovation consists in the operational implementation of
ensemble systems and an accompanying culture change from point
forecasts to distributional forecasts \citep{Leutbecher2008}.  An
ensemble system comprises multiple runs of numerical weather
prediction (NWP) models, where the runs or members differ from each
other in initial conditions and numerical-physical representations of
atmospheric processes.

Ideally, one would like to interpret an ensemble forecast as a random
sample from the conditional distribution of future states of the
atmosphere.  However, this is rarely advisable in practice, as
ensemble forecasts are subject to biases and dispersion errors,
thereby calling for some form of statistical postprocessing
\citep{Gneiting2005a, Vannitsem2018}.  This is typically done by
fitting a distributional regression model, with the weather variable
of interest being the response variable, and the members of the
forecast ensemble constituting the covariates, and applying this model
to future NWP output, to obtain conditional predictive distributions
for future weather quantities.  State of the art techniques include
Bayesian Model Averaging \citep[BMA;][]{Raftery2005, Sloughter2007},
Ensemble Model Output Statistics \citep[EMOS;][]{Gneiting2005b,
  Scheuerer2014}, and Heteroscedastic Censored Logistic Regression
\citep[HCLR;][]{Messner2014}.

In this case study, we apply IDR to the statistical postprocessing of
ensemble forecasts of accumulated precipitation, a variable that is
notoriously difficult to handle, due to its mixed discrete-continuous
character, which requires both a point mass at zero and a right skewed
continuous component on the positive half-axis.  As competitors to
IDR, we implement the BMA technique of \citet{Sloughter2007}, the EMOS
method of \citet{Scheuerer2014}, and HCLR \citep{Messner2014}, which
are widely used parametric approaches that have been developed
specifically for the purposes of probabilistic quantitative
precipitation forecasting.  In contrast, IDR is a generic technique
and fully automatic, once the partial order on the covariate space has
been specified.

\subsection{Data}  \label{subsec:data}

The data in our case study comprise forecasts and observations of
24-hour accumulated precipitation from 06 January 2007 to 01 January
2017 at meteorological stations on airports in London, Brussels,
Zurich and Frankfurt.  As detailed in Table \ref{tab:data}, data
availability differs, and we remove days with missing entries station
by station, so that all types of forecasts for a given station are
trained and evaluated on the same data.  Both forecasts and
observations refer to the 24-hour period from 6:00 UTC to 6:00 UTC on
the following day.  The observations are in the unit of millimeter and
constitute the response variable in distributional regression.  They
are typically, but not always, reported in integer multiples of a
millimeter (mm).

\begin{table} 
\caption{Meteorological stations at airports, with International Air
  Transport Association (IATA) airport code, World Meteorological
  Organization (WMO) station ID, and data availability in days
  (years).  \label{tab:data}}
\small
\begin{tabular}{lccc}
\multicolumn{4}{c}{} \\
\toprule
& IATA Code & WMO ID &  Data Availability \\
\midrule
Brussels, Belgium & BRU & 06449 & 3406 (9.3) \\
Frankfurt, Germany & FRA & 10637 & 3617 (9.9) \\
London, UK & LHR & 03772 & 2256 (6.2) \\
Zurich, Switzerland & ZRH & 06670 & 3241 (8.9) \\
\bottomrule
\end{tabular}
\end{table}

As covariates, we use the 52 members of the leading NWP ensemble
operated by the European Centre for Medium-Range Weather Forecasts
\citep[ECMWF;][]{Molteni1996, Buizza2005}.  The ECMWF ensemble system
comprises a high-resolution member ($\xhres$), a control member at
lower resolution ($\xctr$) and 50 perturbed members ($x_1, \ldots,
x_{50}$) at the same lower resolution but with perturbed initial
conditions, and the perturbed members can be considered exchangeable
\citep{Leutbecher2019}.  To summarize, the covariate vector in
distributional regression is
\begin{equation}  \label{eq:x} 
x 
= \left( x_1, \ldots, x_{50}, \xctr, \xhres \right) 
= \left( \xptb, \xctr, \xhres \right) \in \real^{52},  
\end{equation} 
where $\xptb = (x_1, \ldots, x_{50}) \in \real^{50}$.  At each
station, we use the forecasts for the corresponding latitude-longitude
gridbox of size $0.25 \times 0.25$ degrees, and we consider prediction
horizons of 1 to 5 days.  For example, the two day forecast is
initialized at 00:00 Universal Coordinated Time (UTC) and issued for
24-hour accumulated precipitation from 06:00 UTC on the next calendar
day to 06:00 UTC on the day after.  ECMWF forecast data are available
online via the TIGGE system \citep{Bougeault2010, Swinbank2016}

Statistical postprocessing is both a calibration and a downscaling
problem: Forecasts and observations are at different spatial scales,
whence the unprocessed forecasts are subject to representativeness
error \citep[Chapter~8.9]{Wilks2019}.  Indeed, if we interpret the
predictive distribution from the raw ensemble \eqref{eq:x} as the
empirical distribution of all 52 members --- a customary approach,
which we adopt hereinafter --- there is a strong bias in probability
of precipitation forecasts: Days with exactly zero precipitation are
predicted much less often at the NWP model grid box scale than they
occur at the point scale of the observations.

\subsection{BMA, EMOS and HCLR}  \label{subsec:BMA_EMOS}

Before describing our IDR implementation, we review its leading
competitors, namely, state of the art parametric distributional
regression approaches that have been developed specifically for
accumulated precipitation.

Techniques of ensemble model output statistics
\citep[EMOS;][]{Gneiting2005b} type can be interpreted as parametric
instances of generalized additive models for location, scale and shape
\citep[GAMLSS;][]{Rigby2005}.  The specific variant of
\citet{Scheuerer2014} which we use here is based on the
three-parameter family of left-censored generalized extreme value
(GEV) distributions.  The left-censoring generates a point mass at
zero, corresponding to no precipitation, and the shape parameter
allows for flexible skewness on the positive half-axis, associated
with rain, hail or snow accumulations.  The GEV location parameter is
modeled as a linear function of $\xhres$, $\xctr$, $\mptb =
\frac{1}{50} \sum_{i=1}^{50} x_i$ and
\[
\pzero = \frac{1}{52} 
\left( \one \{ \xhres = 0 \} + \one \{ \xctr = 0 \} + \sum_{i=1}^{50} \one \{ x_i = 0 \} \right),
\]
and the GEV scale parameter is linear in the Gini mean difference
\eqref{eq:Gini} of the 52 individual forecasts in the covariate vector
\eqref{eq:x}.  While all parameters are estimated by minimizing the
in-sample CRPS, the GEV shape parameter does not link to the
covariates.

The general idea of the Bayesian model averaging
\citep[BMA;][]{Raftery2005} approach is to employ a mixture
distribution, where each mixture component is parametric and
associated with an individual ensemble member forecast, with mixture
weights that reflect the member's skill.  Here we use the BMA
implementation of \citet{Sloughter2007} for accumated precipitation in
a variant that is based on $\xhres$, $\xctr$, $\mptb = \frac{1}{50}
\sum_{i=1}^{50} x_i$ only, which we found to achieve more stable
estimates and superior predictive scores than variants based on all
members, as proposed by \citet{Fraley2010} in settings with smaller
groups of exchangeable members.  Hence, our BMA predictive CDF is of
the form
\[
F_x(y) = w_{\rm HRES} \hsp G(y \hsp | \hsp \xhres) 
+ w_{\rm CTR} \hsp G(y \hsp | \hsp \xctr) + w_{\rm PTB} \hsp G(y \hsp | \hsp \mptb)  
\]
for $y \in \real$, where the component CDFs $G(y \hsp | \, \cdot \, )$
are parametric, and the weights $w_{\rm HRES}$, $w_{\rm CTR}$ and
$w_{\rm PTB}$ are nonnegative and sum to one.  Specifically, $G(y \hsp
| \hsp \xhres)$ models the logit of the point mass at zero as a linear
function of $\sqrt[3]{\xhres}$ and $p_{\rm HRES} = \one \{ \xhres =
0 \}$, and the distribution for positive accumulations as a gamma
density with mean and variance being linear in $\sqrt[3]{\xhres}$
and $\xhres$, respectively, and analogously for $G(y \hsp | \hsp
\xctr)$ and $G(y \hsp | \hsp \mptb)$.  Estimation relies on a two-step
procedure, where the (component specific) logit and mean models are
fitted first, followed by maximum likelihood estimation of the weight
parameters and the (joint) variance model via the EM algorithm
\citep{Sloughter2007}.

Heteroscedastic censored logistic regression \citep{Messner2014}
originates from the observation that conditional CDFs can be estimated
by dichotomizing the random variable of interest at given thresholds
and estimating the probability of threshold exceedance via logistic
regression.  The HCLR model used here assumes that square-root
transformed precipitation follows a logistic distribution censored at
zero, with location parameter linear in $\sqrt{\xhres}$,
$\sqrt{\xctr}$ and the mean of the square-root transformed perturbed
forecasts, and a scale parameter linear in the standard deviation of
the square-root transformed perturbed forecasts.  Like EMOS, HCLR can
be interpreted within the GAMLSS framework of \cite{Rigby2005}.

Code for BMA, EMOS and HCLR is available within the {\tt ensembleBMA},
{\tt ensembleMOS} and {\tt crch} packages in \textsf{R}
\citep{Messner2018}.  Unless noted differently, we use default options
in implementation decisions.

\subsection{Choice of partial order for IDR}  \label{subsec:orders} 

IDR applies readily in this setting, without any need for adaptations
due to the mixed-discrete continuous character of precipitation
accumulation, nor requiring data transformations or other types of
implementation decisions.  However, the partial order on the elements
\eqref{eq:x} of the covariate space $\cX = \real^{52}$, or on a
suitable derived space, needs to be selected thoughtfully, considering
that the perturbed members $x_1, \ldots, x_{50}$ are exchangeable.

In the sequel, we apply IDR in three variants. Our first implementation
is based on $\xhres$, $\xctr$ and $\mptb = \frac{1}{50}
\sum_{i=1}^{50} x_i$ along with the componentwise order on $\real^3$,
in that
\begin{equation}  \label{eq:order_cw} 
x \preceq x' \, \iff \, 
\mptb \leq \mptb', \: \xctr \leq \xctr', \: \xhres \leq \xhres'.
\end{equation}
The second implementation uses the same variables and partial order,
but combined with a simple subagging approach: Before applying IDR,
the training data is split into the two disjoint subsamples of
training observations with odd and even indices, and we average the
predictions based on these two subsamples.

Our third implementation combines the empirical increasing convex
order for $\xptb$ with the usual total order on $\real$ for $\xhres$,
whence
\begin{equation}  \label{eq:order_icx} 
x \preceq x' \, \iff \, \xptb \icx \xptb', \: \xhres \leq \xhres'. 
\end{equation} 
Henceforth, we refer to the three implementations based on the partial
orders in \eqref{eq:order_cw} and \eqref{eq:order_icx} as \IDRcw,
\IDRsbg, and \IDRicx.  With reference to the discussion preceding
Theorem \ref{thm:existence}, the relations \eqref{eq:order_cw} and
\eqref{eq:order_icx} define preorders on $\real^{52}$ and partial
orders on $\real^3$ and $\real^{50}_{\uparrow} \times \real$,
respectively.

We have experimented with other options as well, e.g., by
incorporating the maximum $\max_{i = 1, \ldots, 50} x_i$ of the
perturbed members in the componentwise order in \eqref{eq:order_cw},
with the motivation that the maximum might serve as a proxy for the
spread of the ensemble, or by using the empirical stochastic order
$\st$ in lieu of the empirical increasing convex order $\icx$ in
\eqref{eq:order_icx}.  IDR is robust to changes of this type, and the
predictive performance remains stable, provided that the partial order
honors the key substantive insights, in that the perturbed members
$x_1, \ldots, x_{50}$ are exchangeable, while $\xhres$, due to its
higher native resolution, is able to capture local information that is
not contained in $\xptb$ nor $\xctr$.  Hence, $\xhres$ ought to play a
pivotal role in the partial order.

\subsection{Selection of training periods}  \label{subsec:training} 

The selection of the training period is a crucial step in the
statistical postprocessing of NWP output.  Most postprocessing
methods, including the ones used in this analysis, assume that there
is a stationary relationship between the forecasts and the
observations.  As \citet{Hamill2018} points out, this assumption is
hardly ever satisfied in practice: NWP models are updated, instruments
at observation stations get replaced, and forecast biases may vary
seasonally.  These problems are exacerbated by the fact that
quantitative precipitation forecasts require large training datasets
in order to include sufficient numbers of days with non-zero
precipitation and extreme precipitation events.

For BMA and EMOS, a training period over a rolling window of the
latest available 720 days at the time of forecasting is (close to)
optimal at all stations.  This resembles choices made by
\citet{Scheuerer2015} who used a training sample of about 900 past
instances.  \citet{Scheuerer2014} took shorter temporal windows, but
merged instances from nearby stations into the training sets, which is
not possible here.  In general, it would be preferable to select
training data seasonally (e.g., data from the same month), but in our
case the positive effect of using seasonal training data does not
outweigh the negative effect of a smaller sample size.

As a nonparametric technique, IDR requires larger sets of training
data than BMA or EMOS.  As training data for IDR, we used all data
available at the time of forecasting, which is about $2\,500$ to
$3\,000$ days for the stations Frankfurt, Brussels and Zurich, and
$1\,500$ days for London Heathrow.  The same training periods are also
used for HCLR, where no positive effect of shorter, rolling training
periods has been observed \citep{Messner2014}.

For evaluation, we use the years 2015 and 2016 (and 01 January 2017)
for all postprocessing methods and the raw ensemble.  This test
dataset consists of roughly 700 instances for each station and lead
time.

\subsection{Results}  \label{subsec:results} 

Before comparing the BMA, EMOS, \IDRcw , \IDRsbg \ and \IDRicx
\ techniques in terms of out-of-sample predictive performance over the
test period, we exemplify them in Figure \ref{fig:example_CDFs}, where
we show predictive CDFs for accumulated precipitation at Brussels on
December 16, 2015, at a prediction horizon of 2 days.  In panel (a)
the marks at the bottom correspond to $\xhres$, $\xctr$, the perturbed
members $x_1, \ldots, x_{50}$ and their mean $\mptb$.  The observation
at 4 mm is indicated by the vertical line.  Under all four techniques,
the point mass at zero, which represents the probability of no
precipitation, is vanishingly small.  While the BMA, EMOS and HCLR
CDFs are smooth and supported on the positive half-axis, the \IDRcw ,
\IDRsbg \ and \IDRicx \ CDFs are piecewise constant with jump points
at observed values in the training period.  Panel (b) illustrates the
hard and soft constraints on the \IDRcw \ CDF that arise from
\eqref{eq:predbounds} under the order relation \eqref{eq:order_cw},
with the thinner lines representing the \IDRcw \ CDFs of direct
successors and predecessors.  In this example, the constraints are
mostly hard, except for threshold values between 4 and 11 mm.

\begin{figure}[t]
\centering
\includegraphics[width = \textwidth]{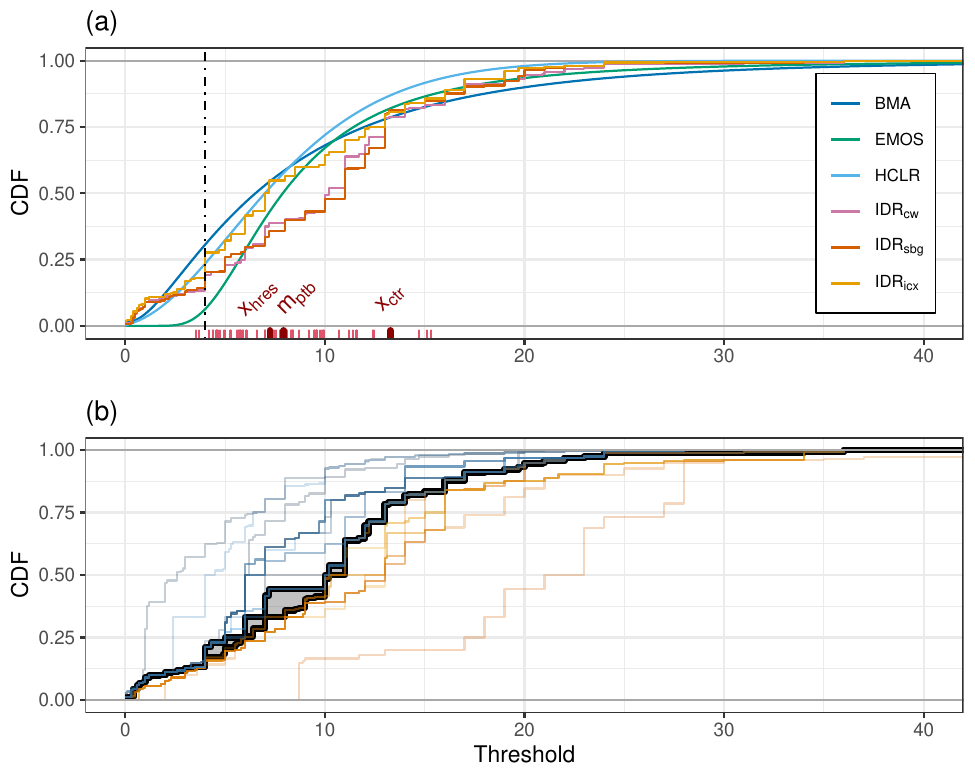}
\caption{Distributional forecasts for accumulated precipitation at
  Brussels, valid 16 December 2015 at a prediction horizon of 2 days.
  (a) BMA, EMOS, \IDRcw, \IDRsbg \ and \IDRicx \ predictive CDFs.  The
  vertical line represents the observation.  (b) \IDRcw \ CDF along
  with the hard and soft constraints in \eqref{eq:predbounds} as
  induced by the order relation \eqref{eq:order_cw}.  The thin lines
  show the \IDRcw \ CDFs at direct predecessors and
  successors.  \label{fig:example_CDFs}}
\end{figure}

\begin{figure}[ht]
\centering
\includegraphics[width = \textwidth]{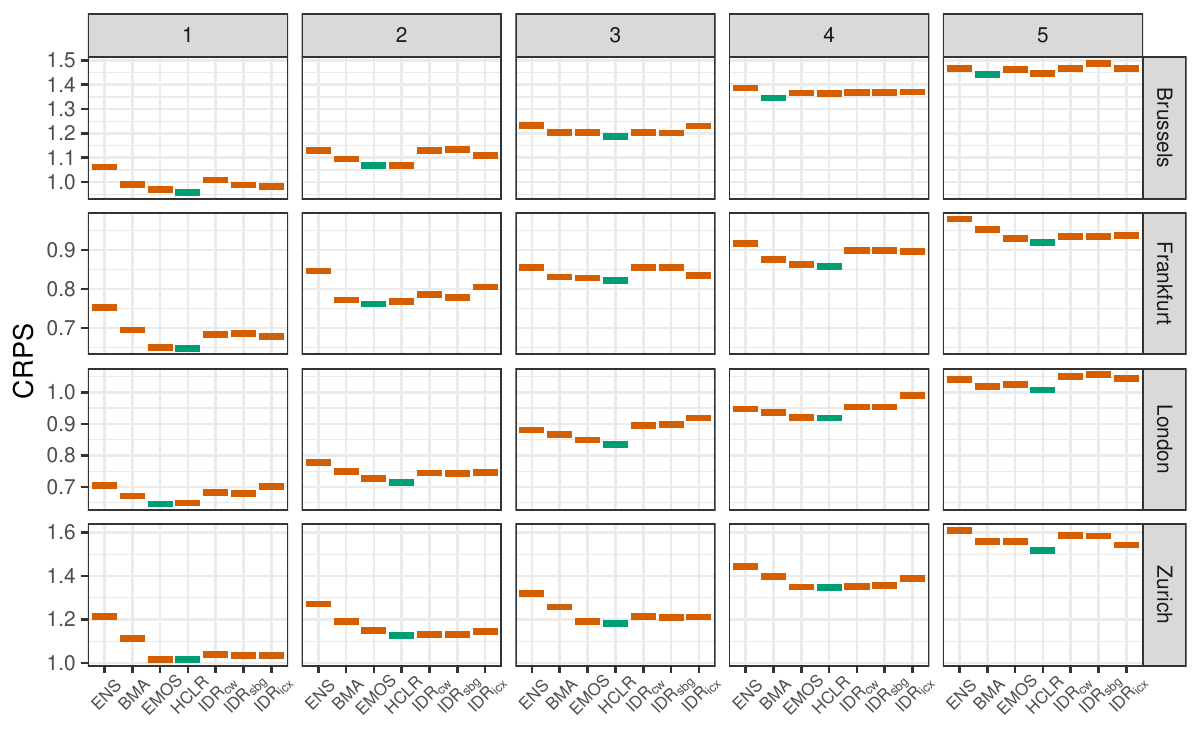}
\caption{Mean CRPS over the test period for raw and postprocessed
  ensemble forecasts of 24-hour accumulated precipitation at
  prediction horizons of 1, 2, 3, 4 and 5 days.  The lowest mean score
  for a given lead time and station is indicated in
  green.  \label{fig:CRPS}}
\end{figure}

\begin{figure}[ht]
\centering
\includegraphics[width = \textwidth]{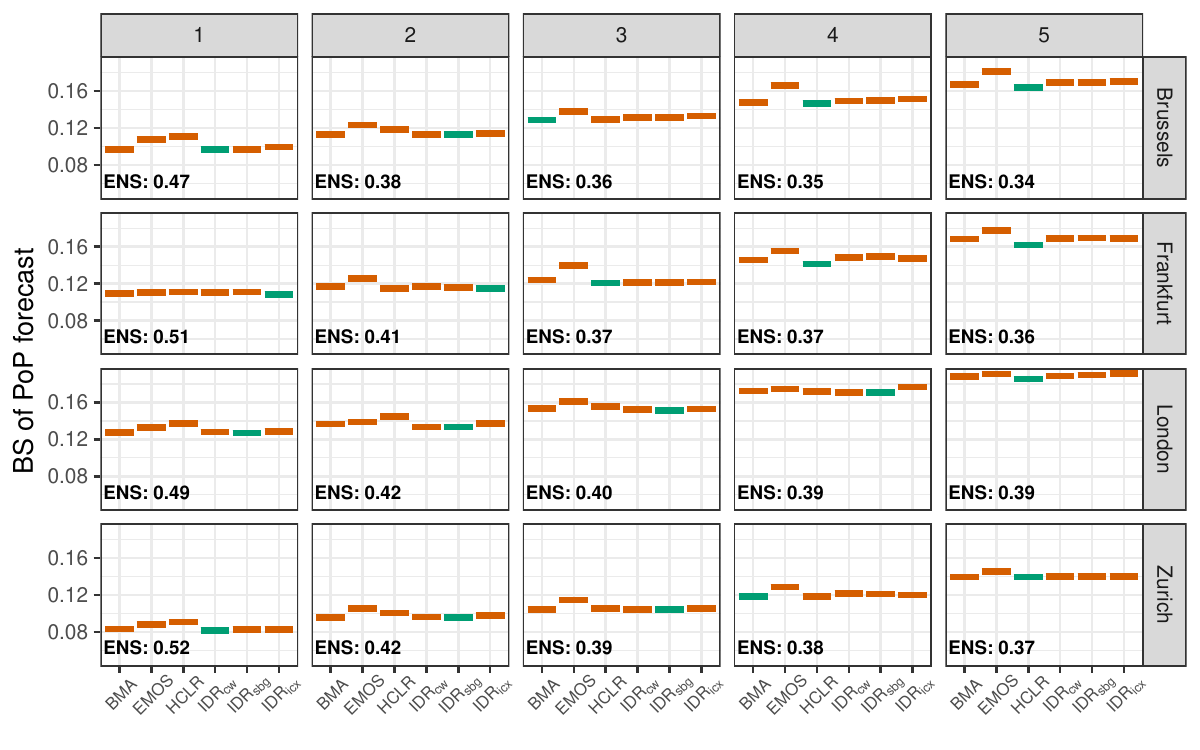}
\caption{Mean Brier score over the test period for probability of
  precipitation forecasts at prediction horizons of 1, 2, 3, 4 and 5
  days.  The lowest mean score for a given lead time and station is
  indicated in green.  \label{fig:BS}}
\end{figure}

We now use the mean CRPS over the test period as an overall measure of
out-of-sample predictive performance.  Figure \ref{fig:CRPS} shows the
CRPS of the raw and postprocessed forecasts for all stations and lead
times, with the raw forecast denoted as ENS.  While HCLR performs best
in terms of the CRPS, the IDR variants show scores of a similar
magnitude and outperform BMA in many instances.  Figure \ref{fig:PIT} in Appendix \ref{app:tab_fig} shows the
difference of the empirical cumulative distribution function (ECDF) of
the PIT defined at \eqref{eq:PIT} to the bisector for the
distributional forecasts.  All three IDR variants show a
PIT-distribution close to uniform, and so do BMA, EMOS and HCLR, as
opposed to the raw ensemble, which is underdispersed.

In Figure \ref{fig:BS} we evaluate probability of precipitation
forecasts by means of the Brier score \citep{Gneiting2007a}, and
Figure \ref{fig:reliability} in Appendix \ref{app:tab_fig} shows reliability diagrams {\green \citep{Wilks2019,
    Dimitriadis2021}}.  As opposed to the raw ensemble forecast, all
distributional regression methods yield reliable probability
forecasts.  BMA, \IDRcw , \IDRsbg \ and \IDRicx \ separate the
estimation of the point mass at zero, and of the distribution for
positive accumulations, and the four methods perform ahead of EMOS.
HCLR is outperformed by BMA and the IDR variants at lead times of one
or two days, but achieves a lower Brier score at the longest lead time
of five days.

Interestingly, IDR tends to outperform EMOS and HCLR for probability
of precipitation forecasts, but not for precipitation accumulations.
We attribute this to the fact that parametric techniques are capable
of extrapolating beyond the range of the training responses, whereas
IDR is not: The highest precipitation amount judged feasible by IDR
equals the largest observation in the training set.  Furthermore,
unlike EMOS and HCLR, IDR does not use information about the spread of
the raw ensemble, which is inconsequential for probability of
precipitation forecasts, but may impede distributional forecasts of
precipitation accumulations.

In all comparisons, the forecast performance of \IDRcw \ and \IDRsbg
\ is similar.  However, in our implementation, the simple subagging
method used in \IDRsbg \ reduced the computation time by up to one
half.

To summarize, our results underscore the suitability of IDR as a
benchmark technique in probabilistic forecasting problems.  Despite
being generic as well as fully automated, IDR is remarkably
competitive relative to state of the art techniques that have been
developed specifically for the purpose.  In fact, in a wide range of
applied problems that lack sophisticated, custom-made distributional
regresssion solutions, IDR might well serve as a ready-to-use,
top-performing method of choice.

\section{Discussion}  \label{sec:discussion} 

\citet{Stigler1975} gives a lucid historical account of the 19th
century transition from point estimation to distribution estimation.
In regression analysis, we may be witnessing what future generations
might refer to as the transition from conditional mean estimation to
conditional distribution estimation, accompanied by a simultaneous
transition from point forecasts to distributional forecasts
\citep{Gneiting2014}.

Isotonic distributional regression (IDR) is a nonparametric technique
for estimating conditional distributions that takes advantage of
partial order relations within the covariate space.  It can be viewed
as a far-reaching generalization of pool adjacent violators (PAV)
algorithm based classical approaches to isotonic (non-distributional)
regression, is entirely generic and fully automated, and provides for
a unified treatment of continuous, discrete and mixed
discrete-continuous real-valued response variables.  Code for the
implementation of IDR within \textsf{R} \citep{R} and Python
(\url{https://www.python.org/}) is available {\green via the {\tt
    isodistrreg} package at CRAN
  (\url{https://CRAN.R-project.org/package=isodistrreg}) and} on github
(\url{https://github.com/AlexanderHenzi/isodistrreg};
\url{https://github.com/evwalz/isodisreg}), with user-friendly
functions for partial orders, estimation, prediction and evaluation.

IDR relies on information supplied by order constraints, and the
choice of the partial order on the covariate space is a critical
decision prior to the analysis.  Only variables that contribute to the
partial order need to be retained, and the order constraints serve to
regularize the IDR solution.  Weak orders lead to increased numbers of
comparable pairs of training instances and predictive distributions
that are more regular.  The choice of the partial order is typically
guided and informed by substantive expertise, as illustrated in our
case study, and it is a challenge for future research to investigate
whether the selection of the partial order could be automated.  Given
that IDR gains information through order constraints, it is a valid
concern whether it is robust under misspecifications of the partial
order.  There is evidence that this is indeed the case: IDR has
guaranteed in-sample threshold calibration (Theorem
\ref{thm:universality}) and therefore satisfies a minimal requirement
for reliable probabilistic forecasts under any (even misspecified)
partial order.  Moreover, \citet[Theorem 7]{ElBarmi2005} show that in
the special case of a discrete, totally ordered covariate, isotonic
regression asymptotically has smaller estimation error than
non-isotonic alternatives even under mild violations of the
monotonicity assumptions, akin to the performance of IDR in the
non-isotonic setting \eqref{eq:sim_noniso} in our simulation study.

Unlike other methods for distributional regression, which require
implementation decisions, such as the specification of parametric
distributions, link functions, estimation procedures and convergence
criteria, to be undertaken by users, IDR is fully automatic once the
partial order and the training set have been identified.  In this
light, we recommend that IDR be used as a benchmark technique in
distributional regression and probabilistic forecasting problems.
With both computational efficiency and the avoidance of overfitting in
mind, IDR can be combined with subsample aggregation (subagging) in
the spirit of random forests.  In our case study on quantitative
precipitation forecasts, we used simplistic ad hoc choices for the
size and number of subsamples.  Future research on computationally
efficient algorithmic implementations of IDR as well as optimal and
automated choices of subsampling settings is highly desirable.

A limitation of IDR in its present form is that we only consider the
usual stochastic order on the space $\cP$ of the conditional
distributions.  Hence, IDR is unable to distinguish situations where
the conditional distributions agree in location but differ in spread,
shape or other regards.  This restriction is of limited concern for
response variables such as precipitation accumulation or income, which
are bounded below and right skewed, but may impact the application of
IDR to variables with symmetric distributions.  In this light, we
encourage future work on ramifications of IDR, in which $\cP$ is
equipped with partial orders other than the stochastic order,
including but not limited to the likelihood ratio order
\citep{Moesching2020a}.  Similarly, the ``spiking'' problem of
traditional isotonic regression, which refers to unwarranted jumps of
estimates at boundaries, arguably did not have adverse effects in our
simulation and case studies.  However, it might be of concern in other
applications, where remedies of the type proposed by \citet{Wu2015}
might yield improvement and warrant study.

Another promising direction for further research are generalizations
of IDR to multivariate response variables.  In weather prediction,
this would allow simultaneous postprocessing of forecasts for several
variables, and an open question is for suitable notions of
multivariate stochastic dominance that allow efficient estimation in
such settings.

\section*{Acknowledgements} 

The authors are grateful to {\green the editor, the associate editor},
four anonymous referees, {\green Sebastian Arnold}, Lutz D\"umbgen, Ed
George, Stephan Hemri, Alexander Jordan, Kristof Kraus, Alexandre M\"osching, Anja
M\"uhle\-mann, Aaditya Ramdas, Nicholas Reich, Benedikt Schulz, Peter
Vogel, {\green Jonas Wallin} and Eva-Maria Walz for providing code,
assistance with the data handling, comments and discussion.  Tilmann
Gneiting acknowledges support by the Klaus Tschira Foundation, by
Deutsche Forschungs\-gemeinschaft (DFG, German Research Foundation) --
Project-ID 257899354 -- TRR 165, and via the Fellowship programme
operated by the European Centre for Medium-Range Weather Forecasts
(ECMWF).  Alexander Henzi and Johanna F.~Ziegel gratefully acknowledge
financial support from the Swiss National Science Foundation.
Calculations were mostly performed on UBELIX
(\url{http://www.id.unibe.ch/hpc}), the HPC cluster at the University
of Bern.

\begingroup
\setstretch{0.95}
\bibliographystyle{chicago}
{\footnotesize \bibliography{biblio}}
\endgroup

\newpage
\hphantom{0}

\vspace{0.2in}
\noindent The appendices are available as Supplementary Material on the JRSSB website; see \url{https://rss.onlinelibrary.wiley.com/doi/full/10.1111/rssb.12450}.
\vspace{0.2in}

\appendix

\section{Proofs for Section \ref{subsec:theory}}  \label{app:theory} 

\begin{proof}[Proof of Theorem \ref{thm:existence}]
	Let $\cA$ be the lattice of all subsets of $\{ 1, \ldots, n \}$ that
	yield admissible superlevel sets for an increasing function $\{ x_1,
	\ldots, x_n \} \to \real$.  More precisely, a set $A \subseteq \{ 1,
	\ldots, n \}$ belongs to $\cA$ if and only if for any $i \in A$ and
	any $x_j$ with $x_i \preceq x_j$ it follows that $j \in A$.
	
	Let $z \in \real$.  By \citet[Theorem 1 and Lemma
	4]{Jordan2021}, the minimizer of the criterion
	\begin{equation}\label{eq:squared}
		\frac{1}{n} \sum_{i=1}^n \left( p_i - \one \{ z \geq y_i \} \right)^2
	\end{equation}
	over all $\bp =(p_1, \ldots, p_n) \in \real^n_{\downarrow, \hsp \bx}$
	is uniquely determined and given by $\bhF(z) = (\hat{F}_1(z), \ldots,
	\hat{F}_n(z)) \in \real^n$ with
	\begin{equation}  \label{eq:exsol}
		\hat{F}_i(z) = \min_{A \in \cA: i \in A} \max_{A' \in \cA : A'
			\subsetneq A} \frac{1}{\#(A \backslash A')} \sum_{j \in A \backslash
			A'} \one \{ y_j \leq z \},
	\end{equation}
	for $i = 1, \ldots, n$, where $\# B$ denotes the cardinality of a set
	$B$.  From the definition of the CRPS it is clear that $\bF$ minimizes
	$\ell_{\mathrm{CRPS}}(\bF)$ over all tuples of functions $\bF = (F_1,
	\ldots, F_n)$ with $F_i : \real \to \real$ such that for each $z \in
	\real$, $(F_1(z), \ldots, F_n(z)) \in \real^n_{\downarrow, \hsp \bx}$.
	It remains to show that for each $i = 1, \ldots, n$, $F_i$ is a valid
	CDF.
	
	Let $i \in \{ 1, \ldots, n \}$, $z \leq z'$, $B \subseteq \{ 1,
	\ldots, n\}$.  It is clear from \eqref{eq:exsol} that the domain of
	$F_i$ in $[0,1]$.  Furthermore,
	\begin{equation}  \label{eq:mono}
		\frac{1}{\# B} \sum_{j \in B} \one \{ y_j \leq z \} \leq \frac{1}{\#
			B} \sum_{j \in B} \one \{ y_j \leq z' \},
	\end{equation}
	and therefore, by \eqref{eq:exsol}, $F_i(z) \leq F_i(z')$.  The
	function $F_i$ is also right-continuous because for $z' \downarrow z$,
	the right-hand side of \eqref{eq:mono} converges to the left-hand
	side.  Finally, for $z \to \pm\infty$ the left-hand side of
	\eqref{eq:mono} converges to zero and one, respectively, which
	concludes the proof.
\end{proof}

\begin{proof}[Proof of Theorem \ref{thm:universality}]
	First, we show threshold calibration.  Let $(X,Y)$ be a random vector
	with distribution $(1/n)\sum_{i=1}^n \delta_{(x_i,y_i)}$ where
	$\delta_{(x_i,y_i)}$ denotes the Dirac measure at $(x_i,y_i)$.  Let $z
	\in \real$.  By \citet[Theorem 6.4]{Lee1983}, there exists a partition
	$\{ B_m \}_{m=1}^M$ of $\{ 1, \dots, n \}$ such that
	\[
	F_i(z) = F_{x_i}(z) = \sum_{m=1}^M \one \{ i \in B_m \}
	\frac{1}{\#B_m} \sum_{j \in B_m} \one \{ y_j \leq z \}.
	\]
	Therefore, the $\sigma$-algebra generated by $F_X(z)$ is contained in
	the $\sigma$-algebra generated by $\{ \bar{B}_m \}_{m=1}^M$ with
	$\bar{B}_m = \{ (x_i,y_i) : \; i \in B_m \}$.  Furthermore,
	\begin{align*}
		\myE \left( \one \{ Y \leq z \} \one \{ (X,Y) \in \bar{B}_m \} \right) 
		& = \frac{1}{n} \sum_{j \in B_m} \one \{ y_j \leq z \} \\
		& = \myE \left( F_X(z) \one \{ (X,Y) \in \bar{B}_m \} \right).
	\end{align*}
	
	Part i) for the scoring rules of type \eqref{eq:CRPS_P} follows
	directly from the arguments in the proof of Theorem
	2.1.  Let $z \in \real$.  By \citet[Theorem 1 and Lemma 4]{Jordan2021} the solution $\bhF(z)$ at
	\eqref{eq:exsol} is not only the unique minimizer of the criterion
	\eqref{eq:squared} but also the unique solution that minimizes
	\begin{equation} \label{eq:criterion2}
		\frac{1}{n} \sum_{i=1}^n 
		\left( \one \{ c < p_i \} - \one \{ y_i \leq z \} \right) 
		\left( c - \one \{ y_i \leq z \} \right)
	\end{equation}
	over all $\bp = (p_1, \ldots, p_n) \in \real^n_{\downarrow, \hsp \bx}$
	simultaneously for all $c \in (0,1)$.  As $\bhF \in \cP^n_{\downarrow,
		\hsp \bx}$, and $(1/n) \sum_{i=1}^n \myS_{z,c}(F_i,y_i)$ is equal to
	the expression at \eqref{eq:criterion2} with $p_i = F_i(z)$, we obtain
	the claim.
	
	Part iii) is a direct consequence of the arguments for the second part
	of part i) and the representation theorem of \citet{Schervish1989} for
	proper scoring rules of binary events.
	
	Let $\alpha \in (0,1)$.  Concerning part ii), observe that any
	function $\mys_\alpha$ satisfying the requirements of the theorem can
	be written as $\int \tilde{\myS}^Q_{\alpha,\theta}(q,y) \diff
	h(\theta)$ for some Borel measure $h$ on $\real$; see \citet[Theorem
	1]{Ehm2016}.  Here,
	\[
	\tilde{\myS}^Q_{\alpha, \hsp \theta}(q,y) = \begin{cases} 
		1 - \alpha, & y \le \theta < q, \\
		\alpha, & q \le \theta < y, \\
		0, & \textrm{otherwise}. \\ 
	\end{cases} 
	\]
	By \citet[Theorem 1 and Proposition 5]{Jordan2021} there exists a unique
	solution $\bhq(\alpha) = (\hat{q}_1(\alpha), \ldots,
	\hat{q}_n(\alpha)) \in \real^n_{\downarrow, \hsp \bx}$ that minimizes
	\[
	\frac{1}{n}\sum_{i=1}^n \tilde{\myS}^Q_{\alpha,\theta}(q_i,y_i)
	\]
	over all $\bq = (q_1, \ldots, q_n) \in \real^n_{\downarrow, \hsp \bx}$
	simultaneously over all $\theta \in \real$ such that for each $i \in
	\{ 1, \ldots, n \}$, $\hat{q}_i(\alpha)$ is the lower
	$\alpha$-sample-quantile of some subset of observations $B_i \subseteq
	\{ y_1, \ldots, y_n \}$.  Indeed, the solution has a max-min
	representation as in \eqref{eq:exsol} with the empirical mean of the
	indcators replaced by the lower $\alpha$-sample quantile over all
	observations in $A \backslash A'$.  The max-min representation for
	$\hat{q}_i(\alpha)$ yields that $\hat{q}_i(\cdot)$ is increasing and
	left-continuous because lower $\alpha$-sample-quantiles are increasing
	and left-continuous as a function of $\alpha$.  Therefore,
	$\hat{q}_i(\cdot)$ is a valid quantile function for each $i = 1,
	\ldots, n$, and the generalized inverse $\bhq^{-1} = (\hat{q}_1^{-1},
	\ldots, \hat{q}_n^{-1})$ is a member of $\cP_{\uparrow, \hsp \bx}^n$.
	
	Since $\myS^Q_{\alpha,\theta}(F,y) =
	\tilde{\myS}^Q_{\alpha,\theta}(F^{-1}(\alpha),y)$ for any CDF $F$, it
	follows from \eqref{eq:CRPS3} that $\bhq^{-1}$ is a CRPS-based
	isotonic regression of $\by$ on $\bx$.  To conclude the proof of part
	ii), it remains to note that $\bhq^{-1} = \bhF$ due to the uniqueness
	of $\bhF$.  The initial statement in part i) is now also immediate.
\end{proof}

\section{Proofs and remarks for Section \ref{subsec:consistency}}  \label{app:consistency} 

The proof of Theorem \ref{thm:consistency} requires the following
lemma, which is established in \citet[Theorem 4.6]{Moesching2020}. 

\begin{lem} \label{lem:upperbound}
	Let $Z_1, Z_2, \ldots$ be independent random variables with
	distribution functions $G_1, G_2, \ldots$, respectively.  For $k = 1,
	2, \ldots$, let
	\[
	\hat{\mathbb{G}}_k(\cdot) = \frac{1}{k} \sum_{i = 1}^k \one\{Z_i \leq \cdot\} 
	\quad \text{and} \quad
	\bar{G}_k (\cdot) = \frac{1}{k} \sum_{i = 1}^k G_i(\cdot).
	\]
	Then there exists a universal constant $M \leq 2^{5/2}e$ such that for all $\eta \geq 0$,
	\[
	\myP \left( \sqrt{k} \, \| \hat{\mathbb{G}}_k - \bar{G}_k \|_{\infty} \geq \eta \right) \leq M \exp(-2\eta^2),
	\]
	where $\|\cdot\|_{\infty}$ denotes the usual supremum norm of functions.
\end{lem}

\begin{proof}[Proof of Theorem \ref{thm:consistency}] 
	Let $\epsilon, \delta > 0$.  By assumption (iv), there exists $r > 0$
	such that
	\begin{equation} \label{eq:unifcontinuity}
		\sup \{| F_x(y) - F_{x'}(y)|: \, x, x' \in [0,1]^d, \|x-x'\| \leq r, \, y \in \real \} < \frac{\epsilon}{4}.
	\end{equation}
	Let $m = \max(\lceil 2/r \rceil, \lceil 2/\delta \rceil + 1)$ and
	define intervals $I_1 = [0, 1/m]$ and $I_j = ((j-1)/m, j/m]$ for $j =
	2, \ldots, m$.  For indices $j_1, \ldots, j_d \in \{ 1, \ldots, m \}$,
	let $I(j_1, \ldots, j_d) = \times_{k = 1}^d I_{j_k} \subset
	[0,1]^d$.  The collection of such rectangles, which we denote by
	$\cR$, partitions $[0,1]^d$ into $m^d$ disjoint subsets with $\sup_{x,
		x' \in I(j_1, \ldots, j_d)} \|x - x'\| \leq r/2$.  
	
	By assumption (i), for each $J \in \cR$, there exists $c_J > 0$ such
	that with asymptotic probability one, $\# (S_n \cap J) \geq n c_J$.
	Define $c = \min_{J \in \cR} c_J > 0$, so that with asymptotic
	probability one, $\# (S_n \cap J) \geq nc > 0$.  We assume in the
	following that for $(X_{n1},Y_{nn}),\dots,(X_{nn},Y_{nn})$ the event
	in assumption (i) occurs for all $J \in \cR$ as well as the event in
	assumption (ii). To ease notation, we drop the subscript $n$.
	
	Let $x = (x_1, \ldots, x_d) \in [\delta,1-\delta]^d$. Then $2/m <
	\delta \leq \min_{i = 1\ldots, d} x_i$ and $\max_{i = 1, \ldots, d}
	x_i \leq 1-\delta < (m-2)/m$, and there exist indices $j_1, \ldots,
	j_d \in \{3, \ldots, m-2\}$ such that $x \in I(j_1, \ldots, j_d)$.
	Define
	\[
	L(x) = I(j_1-1, \ldots, j_d-1), \quad U(x) = I(j_1+1, \ldots, j_d+1).
	\]
	Then $v \preceq x \preceq w$ for all $v \in L(x)$ and $w \in U(x)$,
	and
	\[
	\sup_{v \in L(x)} \|v-x\| \leq r, \quad \sup_{w \in U(x)} \|w-x\| \leq r.
	\]
	We see from \eqref{eq:unifcontinuity} that
	\[
	\sup_{v \in L(x) \cup U(x), y \in \real} |F_{v}(y) - F_x(y)| \leq \frac{\epsilon}{4},
	\]
	whereas the bounds in \eqref{eq:estbounds} give
	\[
	\hat{F}_{X_u}(y) \leq \hat{F}_x(y) \leq \hat{F}_{X_l}(y), \quad 
	y \in \real, \ X_u \in U(x), \ X_l \in L(x).
	\]
	Consequently, for $y \in \real$,
	\begin{align*}
		|\hat{F}_x(y) - F_x(y)| & \leq \max_{j: X_j \in L(x) \cup U(x)} |\hat{F}_{X_j}(y) - F_{X_j}(y)| + \frac{\epsilon}{4} \\
		& \leq \sup_{j: X_j \in (1/m, (m-1)/m]^d, y \in \real} |\hat{F}_{X_j}(y) - F_{X_j}(y)| + \frac{\epsilon}{4},
	\end{align*}
	and this upper bound does not depend on $x$.  Therefore, it is
	sufficient to show that
	\begin{equation}  \label{eq:sufficient} 
		\lim_{n \to \infty} 
		\myP \left( \sup_{j: X_j \in (1/m, (m-1)/m]^d, \, y \in \real} |\hat{F}_{X_j}(y) - F_{X_j}(y)| \geq \frac{3\epsilon}{4} \right) = 0.
	\end{equation}
	
	Let $\cA_n$ be the collection of upper sets in $S_n$.  By the min-max
	formula for antitonic regression, for $j = 1, \ldots, n$ and $y \in
	\real$,
	\[
	\hat{F}_{X_j}(y) = \min_{A \in \cA_n: X_j \in A} \max_{A' \in \cA_n: X_j \not\in A'} 
	\frac{1}{\# (A \setminus A')} \sum_{i: X_i \in A \setminus A'} \one \{ Y_i \leq y \}.
	\]
	For $X_j \in (1/m, (m-1)/m]^d$, let $j_i = \max \{ k : k/m < X_{j,i}
	\} - 1$ and $x_j = (j_1/m, \ldots, j_d/m) \in \real^d$.  Here,
	$X_{j,i}$ denotes the $i$-th component of $X_j$. Then, for all $v
	\in [x_j, X_j] := \{ u \in [0,1]^d : x_j \preceq u \preceq X_j \}$
	it holds that $\|v-X_j\| \leq 2/m \leq r$.  Therefore, inequality
	\eqref{eq:unifcontinuity} along with assumption (iii) imply that for
	all $i$ in $\{1, \ldots, n\}$ such that $X_i \succeq x_j$,
	\[
	F_{X_i}(y) \leq F_{x_j}(y) \leq F_{X_j}(y) + \frac{\epsilon}{4}, 
	\quad y \in \real.
	\]
	Consequently, with $A_j = \{ v \in [0, 1]^d : v \succeq x_j \}$,
	\[
	\hat{F}_{X_j}(y) - F_{X_j}(y) \leq \max_{A' \in \cA_n : X_j \not\in A'}
	\frac{1}{\# (A_j \setminus A')} \sum_{i : X_i \in A_j \setminus A'} \!\! \left( \one \{ Y_i \leq y \} - F_{X_i}(y) \right) 
	+ \frac{\epsilon}{4}.
	\]
	By the definition of $j_1, \ldots, j_d$, $I(j_1+1, \ldots, j_d+1)
	\subseteq [x_j, X_j] \subseteq A_j \setminus A'$ for $A' \in \cA_n$
	with $X_j \not\in A'$.  Therefore, $\# (A_j \setminus A') \geq cn >
	0$, where $c$ is the constant introduced at the beginning of the
	proof.  Lemma \ref{lem:upperbound} implies that for all $A' \subseteq
	A_j$ with $X_j \not\in A'$, conditional on $X_1,\dots,X_n$,
	\[
	\myP \! \left( \sup_{y \in \real} \frac{1}{\#(A_j \setminus A')} 
	\left| \sum_{i: X_i \in A_j \setminus A'} \!\! \left( \one \{ Y_i \leq y \} - F_{X_i}(y) \right) \right| 
	\geq \frac{\epsilon}{2} \right)
	\leq M \exp \! \left( - \frac{c}{2} \epsilon^2 n \right) \! ,
	\]
	with a constant $M \leq 2^{5/2}e$ that does not depend on $j$.  In
	view of the Bonferroni inequality we get the upper bound
	\begin{align*}
		\myP \left( \sup_{y \in \real} \left( \hat{F}_{X_j}(y) - F_{X_j}(y) \right) \geq \frac{3\epsilon}{4} \right) 
		& \leq \sum_{A' \in \cA : X_j \not\in A'} M \exp \left( - \frac{c}{2} \epsilon^2n \right) \\
		& \leq \#(\cA_n) \ M \exp \left( - \frac{c}{2} \epsilon^2n \right), 
	\end{align*}
	which does not depend on $j$.  
	
	For $A \in \cA_n$, let $m(A) = \{ x \in A : z \in A, z \preceq x
	\implies z = x \} \subseteq A$ be the associated set of minimal
	elements.  Then $A = A' \iff m(A) = m(A')$ for $A, A' \in \cA_n$, and
	so the number of upper sets in $S_n$ equals the number of antichains.
	The size of a maximal antichain, which we denote by $s_n$, satisfies
	$s_n \geq 1$ and, by assumption (ii), $s_n \le n^{\gamma}$. So if $n$
	is sufficiently large, $n^{\gamma} < n/2$ and
	\[
	\#(\cA_n) 
	\leq \sum_{k=1}^{s_n} \binom{n}{k} 
	\leq s_n \frac{n!}{(n-s_n)! \, s_n!} 
	\leq \lceil n^{\gamma} \rceil \frac{n!}{(n - \lceil n^{\gamma} \rceil)! \, \lceil n^{\gamma} \rceil!}.
	\]
	By Stirling's formula, the right hand side is asymptotically equivalent to
	\begin{align*}
		n^{\gamma} 
		& \frac{\sqrt{2\pi n} \, (n/e)^n}{\sqrt{2\pi (n - n^\gamma)} \, ((n - n^\gamma)/e)^{n-n^\gamma} 
			\sqrt{2\pi n^\gamma} \, (n^\gamma/e)^{n^\gamma}} \\
		& = \frac{n^{-\gamma/2}}{\sqrt{2\pi (1 - n^{\gamma-1})}} 
		\frac{n^n}{(n - n^\gamma)^{n - n^\gamma} n^{\gamma \, n^\gamma}} \\
		& = \frac{1}{\sqrt{2\pi(1 - n^{\gamma-1})}} n^{-\gamma/2 + n^\gamma(1-\gamma)} (1 - n^{\gamma-1})^{n^\gamma - n} \\ 
		& = \frac{1}{\sqrt{2\pi(1 - n^{\gamma-1})}} 
		\exp \left( \left(- \frac{\gamma}{2} + (1-\gamma) n^{\gamma} \right) \log n \right) (1 - n^{\gamma-1})^{n^\gamma - n},   
	\end{align*}
	where the factor $(1-n^{\gamma-1})^{n^\gamma - n} =
	((1-n^{\gamma-1})^{n^{1 - \gamma}})^{-n^\gamma(1-n^{\gamma-1})}$ grows
	no faster than $\exp(n^{\gamma})$, because $(1-1/x)^x \leq \exp(-1)$
	for $x \geq 1$.  Combining these results, we see that for $n$
	sufficiently large, $\#(\cA_n) \leq \exp(C_1 \, n^\gamma \log n)$,
	where $C_1$ is a constant that depends on $\gamma$.  Hence, for $n$
	sufficiently large,
	\begin{align*}
		\myP \left( \sup_{y \in \real} \left( \hat{F}_{X_j}(y) - F_{X_j}(y) \right) \geq \frac{3\epsilon}{4} \right) 
		& \leq \#(\cA_n) \, M \exp \left( - \frac{c}{2} \epsilon^2 n \right) \\
		& \leq M \exp \left( - \frac{c}{2} \epsilon^2 n + C_1 n^\gamma \log n \right) \\
		& \leq M \exp \left( - C_2 n \right)
	\end{align*}
	for some strictly positive constant $C_2$ that depends on $\gamma$.
	This upper bound does not depend on $j$, so
	\[
	\myP \left( \sup_{j : X_j \in (1/m, (m-1)/m]^p, \, y \in \real} \left( \hat{F}_{X_j}(y) - F_{X_j}(y) \right) \geq \frac{3\epsilon}{4} \right) 
	\leq M \exp \left( - C_2 n \right) n 
	\]
	vanishes as $n \to \infty$.  Analogous arguments yield the bound with
	$F_{X_j}$ and $\hat{F}_{X_j}$ interchanged, which establishes
	\eqref{eq:sufficient} and completes the proof.
\end{proof}

As noted, the broad applicability of Theorem \ref{thm:consistency}
rests on a powerful combinatorial result of \citet[Corollary
2]{Brightwell1992}, which enables us to deduce consistency without
having to check complex regularity conditions of the type in
\citet{Robertson1975}.  The size of a maximal antichain also appears
in the derivation of risk bounds for multiple isotonic regression for
the mean in \citet[p. 2447, and Lemma 4 in \green{their} Supplementary
Material]{Han2019}.  Their Lemma 4 gives an asymptotic lower bound
of $n^{1-1/d}$ for the size of a maximal antichain among $n$
independent and identically distributed covariates $X_1, \ldots, X_n
\in \real^d$ with any Lebesgue density bounded from above, and might
in fact also be derived from \citet[Corollary 2]{Brightwell1992}.  An
intuitive explanation for the lower bound $n^{1-1/d}$ is that any
distribution with bounded Lebesgue density can be restricted to a
fixed subset where the density is positive, and asymptotically the
maximum antichain of $X_1, \ldots, X_n$ within this subset behaves as
if $X_i \sim \operatorname{Unif}[0,1]^d$, regardless of the dependence
structure.  This is an interesting result, because if the speed of
convergence hinges on the maximal size of an antichain, as our proof
and results in \citet{Han2019} suggest, then it may not be possible to
improve the speed of convergence by assuming positively correlated
components.  Therefore, we believe that positive dependency between
the components of the covariate vector does not affect convergence
rates, though clearly it may have positive effects in finite sample
settings.

\section{Proofs for Sections \ref{subsec:st} and \ref{subsec:icx}}  \label{app:orders} 

\begin{proof}[Proof of Proposition \ref{prop:st}]
	Denote the CDF corresponding to the empirical distribution of $x_1,
	\ldots, x_d$ and of $x'_1, \ldots, x'_d$ by $F$ and $G$, respectively.
	For part i), assume that $x_{(i)} \leq x'_{(i)}$ for $i = 1, \ldots,
	d$, and let $z \in \real$. Then,
	\[
	F(z) = \frac{\# \{ i : x_{(i)} \leq z \}}{d} \geq \frac{\# \{ i : x'_{(i)} \leq z \}}{d} = G(y),
	\]
	hence $F$ is smaller than $G$ in the usual stochastic order.
	Conversely, if $F$ is smaller then $G$, by choosing $z = x'_{(k)}$, $k
	= 1, \ldots, d$, we obtain
	\[
	\frac{\# \{ i : x_{(i)} \leq x'_{(k)} \}}{d} = F(x'_{(k)}) \geq G(x'_{(k)}) = \frac{ \# \{ i : x'_{(i)} \leq x'_{(k)} \}}{d}.
	\]
	By definition of the $k$-th order statistic, we know that $\# \{ i :
	x'_{(i)} \leq x'_{(k)} \} \geq k$ (with equality if the $x'_i$ are
	distinct).  Therefore, $\# \{ i : x_{(i)} \leq x'_{(k)} \} \geq k$.
	This can only be true if $x_{(k)} \leq x'_{(k)}$.
	
	Concerning part ii), we can assume without loss of generality that
	$x_1 \leq x_2 \leq \cdots \leq x_d$, otherwise we reorder the pairs
	$(x_i, y_i)$.  Now apply part i): We know that $x_1 \leq x'_1$ and
	$x'_{(1)} \geq x_j$ for some $j$.  But the components of $x$ are
	sorted, hence $x'_{(1)} \geq x_{j} \geq x_1 = x_{(1)}$, and also $x'_1
	\geq x'_{(1)} \geq x_j$.  So we can think of the positions of $x'_1$
	and $x'_{(1)}$ in $x'$ to be exchanged, without violating the
	condition $x \preceq x'$.  Now we can ignore the pair $(x'_1,
	x'_{(1)})$ and proceed in the same way for remaining components
	$(x_i)_{i=2}^d$ and $(x'_i)_{i=2}^d$.
	
	For the proof of part iii), assume the opposite, i.e., $x_i \geq x'_i$
	for $i = 1, \ldots, d$.  By ii), we know that $x \succeq_{\textrm{st}}
	x'$.  By assumption $x \st x'$, hence $x$ and $x'$ are permutations of
	each other.  But then either $x = x'$, or $x$ and $x'$ cannot be
	comparable in the componentwise order.
	
	The last part is immediate from part i).  
\end{proof}

\begin{proof}[Proof of Proposition  \ref{prop:icx}]
	Part i) is a consequence of Theorem 4.A.3 of \citet{Shaked2007}.  Part
	ii) follows from part i) and Proposition \ref{prop:st} i).  For part
	iii) note that the Gini mean difference has the equivalent formula
	\[
	g(x) = \frac{2}{d(d-1)} \sum_{i = 1}^d x_{(i)}(2i - d - 1),
	\]
	which can be rewritten as
	\[
	g(x) = \frac{4}{d(d-1)} \sum_{i = 1}^d \sum_{j = i}^d x_{(j)} - 2\frac{d + 1}{d(d - 1)}\sum_{i=1}^d x_i.
	\]
	Part i) implies that
	\begin{align*}
		g(x') +  2 \frac{d + 1}{d(d - 1)} \sum_{i=1}^d x_i' 
		& = \frac{4}{d(d-1)} \sum_{i = 1}^d \sum_{j = i}^d x_{(j)}'  \\
		& \geq \frac{4}{d(d-1)} \sum_{i = 1}^d \sum_{j = i}^d x_{(j)}  
		= g(x) +  2\frac{d + 1}{d(d - 1)}\sum_{i=1}^d x_i. \qedhere
	\end{align*}
\end{proof}

\section{Large sample equivalence of CRPS and \protect\boldmath $L_2$ measures}  \label{app:CRPS.L2} 

Here we show that the difference between the mean CRPS for the
distributional regression method at hand and the mean CRPS for the
ideal forecast is large sample equivalent to the (squared) $L_2$ error
in conditional distribution estimation.  This relates the CRPS, as
introduced by \citet{Matheson1976} and arguably the most prevalent
measure of predictive performance in distributional forecasting
\citep{Gneiting2007a}, to traditional $L_p$ measures, as used by
\citet{Hall1999} and \citet{Spady2018} in the evaluation of
conditional cumulative distribution function (CDF) estimation.

Specifically, suppose that the random variates $(x_1, y_1), \ldots,
(x_m, y_m)$ are independent identically distributed from a population
with bivariate law $G$.  Let $F(Y|X)$ be any estimate of the
conditional distributions of $Y$ given $X$, and for $i = 1, \ldots, m$
let $F_i = F( Y \mid X = x_i)$ and $G_i = G( Y \mid X = x_i)$ denote
the respective conditional CDFs for $x_1, \dots, x_m$.  Subject to the
conditions of the bivariate strong law of large numbers,
\[
\bar{\myS}_m^F = \frac{1}{m} \sum_{i=1}^m \textrm{CRPS}(F_i,y_i) 
\to \myE_{(X,Y) \sim G} \left[ \textrm{CRPS}(F(Y|X),Y) \right]
\]
and 
\[  
\bar{\myS}_m^G = \frac{1}{m} \sum_{i=1}^m \textrm{CRPS}(G_i,y_i) 
\to \myE_{(X,Y) \sim G} \left[ \textrm{CRPS}(G(Y|X),Y) \right]
\]
almost surely.  Therefore, subject to the conditions of the strong law
and Fubini's theorem,
\begin{align*}  
	\bar{\myS}_m^F - \bar{\myS}_m^G 
	& \to \myE_{X \sim G} \, \myE_{Y \sim G(Y \mid X)} 
	\left[ \textrm{CRPS}(F(Y|X),Y) - \textrm{CRPS}(G(Y|X),Y) \mid X \right] \\
	& = \myE_{X \sim G} \left[ \int_{-\infty}^\infty \! \left( F(y \mid X) - G(y \mid X) \right)^2 \textrm{d}y \right] \\
	& = \myE_{X \sim G} \left[ L_2^2 \left( F( \cdot \mid X), G( \cdot \mid X) \right) \right] 
\end{align*} 
almost surely, where the first equality uses the analytic form of the
CRPS divergence \citep[p.~367]{Gneiting2007a}.

In the context of the simulation study in Section
4, the above setting corresponds to a single
of the 500 Monte Carlo replicates, where $F$ is an estimate on a
training set of size $n$, and performance is evaluated on an
independent test sample of size $m = 5\,000$.  The large sample
arguments remain valid when scores are averaged across Monte Carlo
replicates.

\section{Additional tables and figures} \label{app:tab_fig}

Table \ref{tab:sim_description} provides implementation details
{\green for} the distributional regression methods {\green in} the
simulation study {\green in Section \ref{sec:simulation_study}}.

\begin{table}
	\caption{Implementation details for the distributional regression
		methods in the simulation study.  We list the \textsf{R} packages
		and the specific functions used for estimation and prediction, along
		with choices for tuning parameters.  For nonparametric kernel
		smoothing (NP) we use Gaussian kernels in \eqref{eq:sim_standard},
		\eqref{eq:sim_discontinuous}, and \eqref{eq:sim_noniso} and the {\tt
			liracine} kernel in the Poisson scenario \eqref{eq:sim_poisson}.
		To fit semiparametric quantile regression (SQR) and conditional
		transformation models (TRAM) we employ cubic $B$-splines with
		interior knots from 2 to 8 in steps of 2 and boundary knots 0 and 10
		({\tt bs(x, \ldots)}).  For TRAM, we use continuous outcome logistic
		regression ({\tt Colr}) for \eqref{eq:sim_standard},
		\eqref{eq:sim_discontinuous}, and \eqref{eq:sim_noniso}, and ordered
		categorical regression ({\tt Polr}) in \eqref{eq:sim_poisson}.  For
		further detail, see the code, which is available at
		\url{https://github.com/AlexanderHenzi/isodistrreg}.
		\label{tab:sim_description} \bigskip}
	
	\renewcommand{\arraystretch}{1.1}
	\resizebox{\textwidth}{!}{
	\begin{tabular}{l|l}
		\toprule
			  \multicolumn{2}{l}{\hspace{0.5cm} Package}  \\
		NP & {\tt np} \citep{Hayfield2008} \\
		SQR & {\tt quantreg} \citep{Koenker2020}\\
		TRAM & {\tt tram} \citep{Hothorn2020} \\
		QRF &  {\tt grf} \citep{Tibshirani2020}\\
		IDR & {\tt isodistrreg} \\
		\IDRsbg & {\tt isodistrreg} \\
		\midrule
			  \multicolumn{2}{l}{\hspace{0.5cm} Estimation}  \\
		NP & {\tt npcdistbw(nmulti = 4, oykertype = "liracine", bwtype = adaptive\_nn")} \\
		SQR & {\tt rq(y$\sim$., data = cbind(y = y, bs(x, \ldots)), tau = seq(0.005,0.995,0.001))} \\
		TRAM & {\tt Colr/Polr(y$\sim$., data = cbind(y = y, bs(x, \ldots)))} \\
		QRF & {\tt quantile\_forest(min.node.size = 40, quantiles = seq(0.01,0.99,0.01))} \\
		IDR & {\tt idr()} \\
		\IDRsbg & {\tt idrbag(b = 100, digits = 6, p = 1/2)} \\
		\midrule 
			  \multicolumn{2}{l}{\hspace{0.5cm} Prediction}  \\
		NP & {\tt npcdist(eydat = grid)} \\
		SQR & {\tt predict.rqs()} \\
		TRAM & {\tt predict.ctm(K = 5000,\ type = "distribution")} \\
		QRF & {\tt predict.quantile\_forest(quantiles = seq(0.005,0.995,0.001))} \\
		IDR & {\tt predict.idrfit(digits = 6)} \\
		\IDRsbg & {\tt idrbag(b = 100, digits = 6, p = 1/2)} \\
		\bottomrule
	\end{tabular}
	}
\end{table}

Figure \ref{fig:PIT} assesses the probabilistic calibration of the
postprocessing methods for precipitation forecasts in the case study
{\green in Section \ref{sec:case_study}.  Similarly,} Figure
\ref{fig:reliability} shows reliability diagrams for the postprocessed
probability of precipitation forecasts, {\green using the CORP
	approach developed} by \citet{Dimitriadis2021}.  

\begin{sidewaysfigure}[p]
	\includegraphics[width = 0.95 \textwidth]{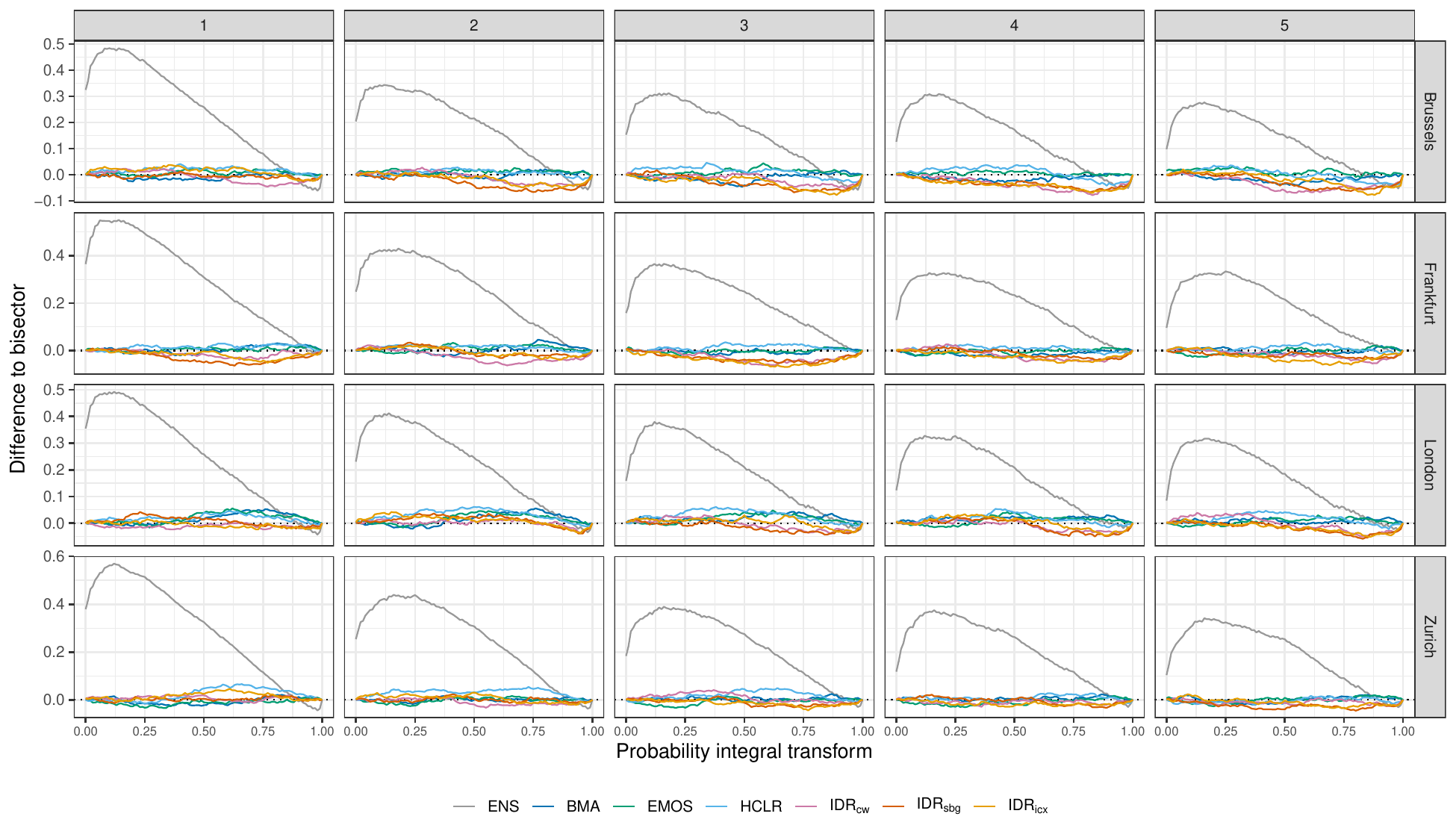}
	\caption{Difference of the ECDF of the PIT to the bisector (positive:
		above bisector) for raw and postprocessed ensemble forecasts of
		24-hour accumulated precipitation at prediction horizons of 1, 2, 3,
		4 and 5 days, for the test period.  \label{fig:PIT}}
\end{sidewaysfigure}

\begin{sidewaysfigure}[p]
	\includegraphics[width = 0.95 \textwidth]{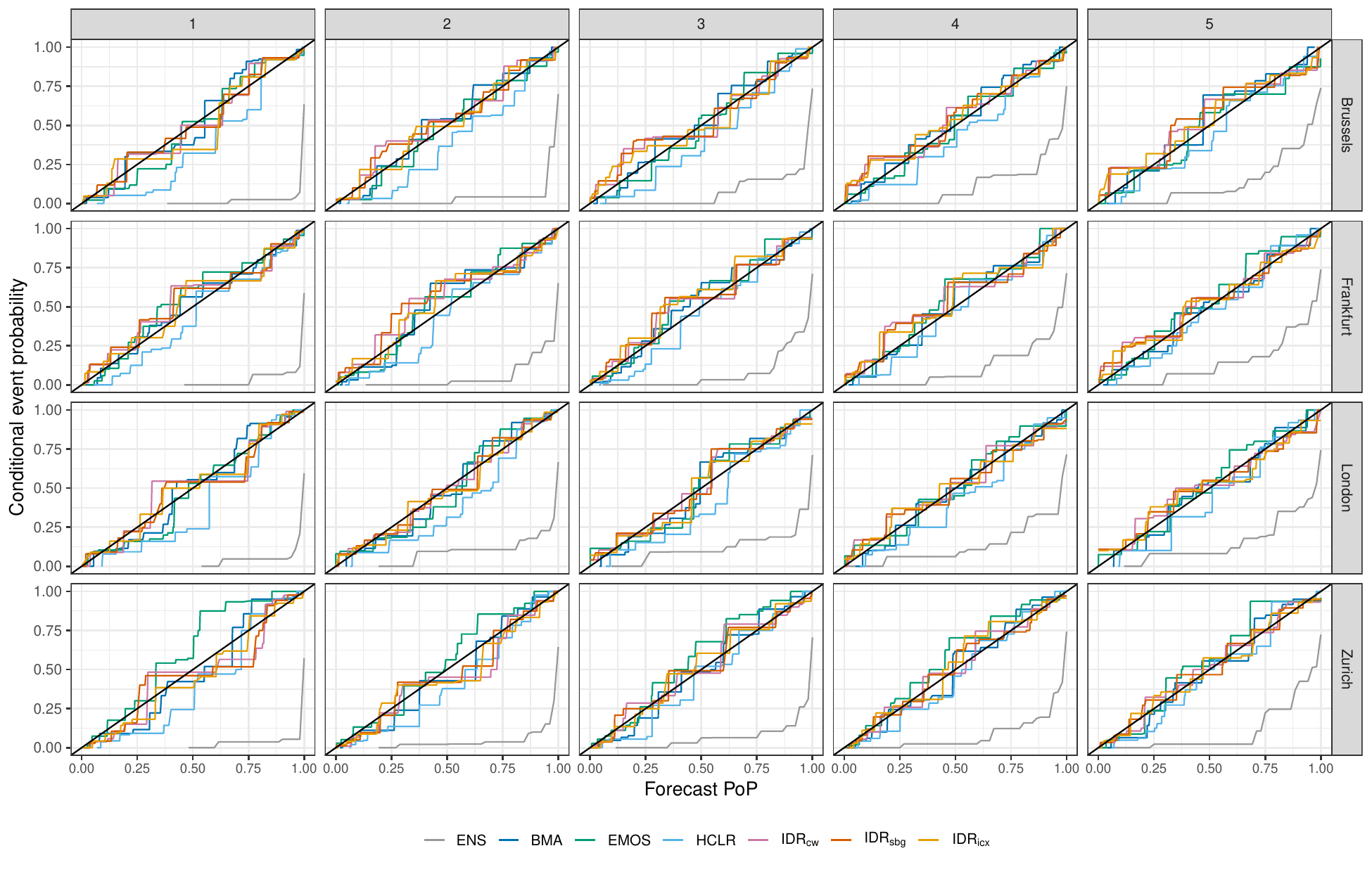}
	\caption{CORP reliability diagrams for probability of precipitation
		forecasts at prediction horizons of 1, 2, 3, 4 and 5 days, for the
		test period.  \label{fig:reliability}}
\end{sidewaysfigure}

\end{document}